\documentclass[cmp,numbook,envcountsame]{svjour}

\usepackage{amsmath} 
\usepackage{amssymb}
\usepackage{amsfonts}

\usepackage{cite}

\newcommand{\BB}{{\mathbb B}}
\newcommand{\II}{{\mathbb I}}

\newcommand{\CC}{{\mathbb C}}
\newcommand{\RR}{{\mathbb R}}

\newcommand{\NN}{{\mathbb N}}

\newcommand{\ZZ}{{\mathbb Z}}

\newcommand{\cB}{{\mathcal{B}}}
\newcommand{\cD}{{\mathcal{D}}}

\newcommand{\cF}{{\mathcal{F}}}
\newcommand{\cH}{{\mathcal{H}}}

\newcommand{\cK}{{\mathcal{K}}}

\newcommand{\bia}{\mbox{\boldmath $a$}}
\newcommand{\bib}{\mbox{\boldmath $b$}}
\newcommand{\bix}{\mbox{\boldmath $x$}}
\newcommand{\sbix}{{\mbox{\footnotesize \boldmath $x$}}}
\newcommand{\biy}{\mbox{\boldmath $y$}}
\newcommand{\sbiy}{{\mbox{\footnotesize \boldmath $y$}}}

\newcommand{\biK}{\mbox{\boldmath $K$}}
\newcommand{\biO}{\mbox{\boldmath $O$}}
\newcommand{\biP}{\mbox{\boldmath $P$}}
\newcommand{\sbiP}{\mbox{\footnotesize \boldmath $P$}}
\newcommand{\bip}{\mbox{\boldmath $p$}}
\newcommand{\biQ}{\mbox{\boldmath $Q$}}
\newcommand{\biq}{\mbox{\boldmath $q$}}
\newcommand{\biR}{\mbox{\boldmath $R$}}

\newcommand{\biV}{\mbox{\boldmath $V$}}
\newcommand{\one}{\mbox{\boldmath $1$}}

\newcommand{\bfA}{\mbox{\boldmath $\mathfrak A$}}
\newcommand{\obfA}{\overline{\bfA}}
\newcommand{\fC}{{\mathfrak C}}

\newcommand{\fF}{{\mathfrak F}}

\newcommand{\fJ}{{\mathfrak J}}
\newcommand{\bfJ}{\mbox{\boldmath $\mathfrak J$}}
\newcommand{\fK}{{\mathfrak K}}
\newcommand{\bfK}{\mbox{\boldmath $\mathfrak K$}}

\newcommand{\fR}{{\mathfrak R}}
\newcommand{\bfR}{\mbox{\boldmath $\mathfrak R$}}

\newcommand{\balpha}{\mbox{\boldmath $\alpha$}}

\newcommand{\bdelta}{\mbox{\boldmath $\delta$}}
\newcommand{\bDelta}{\mbox{\boldmath $\Delta$}}

\newcommand{\bPsi}{\mbox{\boldmath $\Psi$}}
\newcommand{\bPhi}{\mbox{\boldmath $\Phi$}}
\newcommand{\bSigma}{\mbox{\boldmath $\Sigma$}}
\newcommand{\sbSigma}{\mbox{\footnotesize \boldmath $\Sigma$}}

\newcommand{\bpartial}{\mbox{\boldmath $\partial$}}

\newcommand{\scirc}{\mbox{\footnotesize $\circ$}}

\def\eg{{\it e.g.\ }}
\def\ie{{\it i.e.\ }}
\def\viz{{\it viz.\ }}

\title{The resolvent algebra of non-relativistic Bose fields: 
observables, dynamics and states \\[3mm] 
{\normalsize Dedicated to Klaus Fredenhagen on his seventieth 
birthday}}

\author{Detlev Buchholz}
\institute{Mathematisches Institut, 
Universit\"at G\"ottingen, \\ 37073 G\"ottingen, Germany}

\authorrunning{Detlev Buchholz} 
\titlerunning{The resolvent algebra of non-relativistic Bose fields}

\date{}

\begin{document}

\maketitle

\begin{abstract} 
\noindent
The structure of the gauge invariant (particle number preserving) 
C*-algebra generated by the resolvents of a non-relativistic 
Bose field is analyzed. It is shown to form a dense subalgebra of 
the bounded inverse limit of a directed system of approximately finite 
dimensional C*-algebras. Based on this observation, it is proven 
that the closure of the gauge invariant algebra is stable 
under the dynamics induced by Hamiltonians involving pair potentials. 
These facts allow to proceed to a description of interacting Bosons 
in terms of C*-dynamical systems. 
It is outlined how the present approach leads to simplifications in 
the construction of infinite bosonic states and sheds new light 
on topics in many body theory. 
\end{abstract}

\section{Introduction}
\setcounter{equation}{0}

We continue in this article our study of the resolvent algebras 
of canonical quantum systems, which 
provide a natural C*-algebraic framework for 
the construction of non-trivial dynamics and of the corresponding states 
\cite{BuGr2, BuGr3, Bu1}. Moreover, they encode 
characteristic kinematical properties of the underlying   
finite or infinite quantum systems  
and have an intriguing algebraic structure~\cite{Bu2}. 
In the present article,  we supplement these results by a 
study of the resolvent algebra of a scalar non-relativistic 
quantum field, satisfying canonical commutation relations. 
This framework provides an efficient alternative to the particle 
picture, based on position and momentum operators, which is  
frequently used in the analysis of bosonic many body systems. 
Let us briefly recall some facts supporting this view.

First, the particle picture generically looses its significance 
when one is dealing with infinite systems, such as in case of 
the thermodynamic limit of equilibrium states or in the presence of an 
infinity 
of low energy excitations (infrared clouds) in states of finite energy. 
Strictly speaking, even in systems with a finite particle number
the particle interpretation of states acquires 
significance only at asymptotic times when the interaction between 
subsystems fades away and particle features emerge. The field 
theoretic formalism provides a complementary point of view. It is based 
on the concept of localized operations and observables and thereby 
tries to model what actually happens in laboratories. Moreover, 
it simultaneously covers finite and infinite systems and thus 
provides a uniform basis for their analysis and interpretation. 

Second, in order that some  
algebra generated by field operators may be regarded as a suitable 
kinematical  framework for the formulation of dynamics, 
it ought to incorporate the solutions of the Heisenberg equations  
for a large class of Hamiltonians of physical 
interest. If one considers 
the polynomial algebra generated by canonical field operators, 
one finds, however, that it is stable only under a small family of 
rather trivial dynamics (inducing symplectic transformations).
It is less known that this is also the case for the 
Weyl algebra of exponentials of the fields \cite{BuGr2,FaVe}. 
The resolvent algebra is much better behaved in this 
respect. As we shall show, a slight extension of 
its gauge invariant (particle number preserving)
subalgebra is stable under the dynamics induced by 
a large family of Hamiltonians describing pair interactions. 
Thus it comprises kinematical operators which can be   
used at any time to describe the underlying system. In simple words: 
the algebra does not only contain the initial values 
but also the solutions of the Heisenberg equations. 

Third, in the analysis of infinite systems one frequently relies 
on finite volume approximations (boxes). This 
approach requires the consideration of boundary conditions and 
modifications of kinematical algebras if one wants to proceed 
to the thermodynamic limit; it is a somewhat 
cumbersome procedure. In the field theoretic setting of the resolvent
algebra one deals from the outset with infinite space. Trapped systems
can be described nevertheless by adding to the Hamiltonians   
external confining potentials. The resulting dynamics still act on 
the extended algebra. Thus the resolvent algebra 
provides a convenient framework for the study of infinite systems
and their approximations. 

In the case considered here, the resolvent algebra is generated by 
the resolvents of a quantum field which is 
formed by linear combinations of 
creation and annihilation operators in $s$ spatial 
dimensions. 
These operators satisfy standard canonical 
commutation relations in position space,
informally given by 
$$
[a(\bix), a^*(\biy)] = \delta(\bix - \biy) \, 1 \, , \quad
[a(\bix), a(\biy)] = [a^*(\bix), a^*(\biy)] = 0 \, , \quad 
\bix, \biy \in \RR^s \, . 
$$
The advantage gained by using the resolvents of field operators,
smoothed with test functions, rests upon the fact that, in contrast 
to the exponential Weyl operators, large values of the fields 
are suppressed from the outset. The expectation values of the 
corresponding resolvents simply vanish in singular states. 
Thereby the apparent obstructions to a C*-algebraic
treatment of interacting bosonic systems, 
envisaged for example in \cite[Sec.~6.3]{BrRo} and \cite{NaTh}, 
become irrelevant.  
The simplifications, which arise by using the resolvents, have
not yet been exploited in the literature, 
to the best of our knowledge. Furthermore, this framework 
can be applied to an arbitrary number and 
arbitrary types of Bose fields.

As already mentioned, we will consider pair interactions of the
field. The generators of the time translations thus have 
on their standard domains of definition the form 
\begin{equation} \label{e1.1}
H = \int \! d\bix \, \bpartial a^*(\bix) \, \bpartial a(\bix) + 
\int \! d\bix \! \int \! d\biy \ a^*(\bix) a^*(\biy) \, V(\bix - \biy) \,
a(\bix) a(\biy) \, .  
\end{equation}
Here $\bpartial$ denotes the gradient with regard
to $\bix$ and we assume for simplicty that the potential 
$V$ is a real, continuous and symmetric function which vanishes 
at infinity. Singular and unbounded 
potentials can be treated in a similar manner, 
cf.\ \cite[Sec.~6]{Bu1}, but they are not considered here. 
We will make substantial use of the fact that $H$ commutes 
with the particle number operator, which is given by 
$N = \int \! d\bix \, a^*(\bix) a(\bix)$. 
 
In the subsequent section we will explicate these structures 
in more precise mathematical terms. We will make use of the fact
that the (abstractly defined) C*-algebra,   
describing the resolvents of a canonical quantum field, smeared with 
test functions, is faithfully represented in any regular 
representation, such as the Fock representation \cite[Thm.~4.10]{BuGr2}.
We will deal with the resolvent algebra in this fixed 
representation, denoting it by $\bfR$, since we can take 
advantage there of the simple structure of the 
underlying states. 

We focus in this article on the gauge invariant (particle number 
preserving) subalgebra \mbox{$\bfA \subset \bfR$} and determine its structure.
This algebra is equipped with a C*-norm and an increasing family
of C*-seminorms defining the norm in the limit. It will be convenient 
in the construction of dynamics to complete the algebra 
$\bfA$ in the topology induced by the seminorms, 
leading to a slight extension $\overline{\bfA} \supset \bfA$.
The algebra $\overline{\bfA}$ 
is again a C*-algebra. The two algebras coincide (as sets)
on all subspaces of Fock space with finite particle number,
\ie the extension becomes visible only in states describing an infinity
of particles. We will show that the algebra $\overline{\bfA}$,
henceforth called observable algebra, has 
a comfortable mathematical structure: it is isomorphic to the (bounded) 
inverse limit of a directed system of approximately finite dimensional 
C*-algebras. Note that this type of algebras goes 
by differing names in the 
literature, cf.~\cite{Ph} and references quoted there.

These observations facilitate the construction of dynamics of the 
observable algebra $\overline{\bfA}$.
Relying on arguments given in~\cite{Bu2,BuGr3}, we will show that 
the unitary time evolution operators $e^{itH}$, which are determined by 
the above Hamiltonians on Fock space, induce by their adjoint action 
$\balpha(t) \doteq \mbox{Ad} \, e^{itH}$, $t \in \RR$, 
automorphisms of the faithfully represented algebra $\overline{\bfA}$. 
Moreover, this action is pointwise continuous 
in the topology given by the 
seminorms. This fact allows to proceed from $\overline{\bfA}$ to 
C*-dynamical subsystems $(\bfA_\alpha, \balpha)$
for any given dynamics~$\balpha$. The dynamics 
also preserves the asymptotic commutativity of 
observables at large spatial distances, thereby complying 
with a principle of kinematical
causality. The existence of an important family of dynamics 
is thereby established, showing that the algebra of observables 
$\overline{\bfA}$ is an acceptable kinematical algebra
in the sense explained above. 

Having settled the framework, we will outline some  
applications of the present formalism to the construction of 
states and their interpretation.
The topics covered are ground states, approximate
ground states consisting of an infinity of low energy Bosons
(infra vacua) and the recovery of particle 
observables in the field theoretic setting at asymptotic times. 
We also indicate how thermal equilibrium states can be 
constructed in this setting by considering Gibbs-von 
Neumann ensembles for Hamiltonians with some additional confining
potential. Turning off these external potentials, the automorphic 
action of the dynamics converges on the observables 
to its original form and, by compactness arguments, one obtains 
limit states on the algebra of observables, describing systems 
in the thermodynamic limit. The mathematical arguments entering in 
these constructions are based on familiar methods in the theory of 
operator algebras and are largely omitted. 

Our article is organized as follows. In the subsequent section we 
recall some definitions and 
facts about the resolvent algebra, which are used in the 
present invesitigation. Section~3 contains
the structural analysis of the algebra of gauge invariant operators
generated by the resolvents of a Bose field. 
The automorphic action of the dynamics on the algebra of 
observables is established in Sec.~4. 
In Sec.~5 some topics in many body theory are addressed, 
whereby certain technical points are deferred to an appendix.  
The article closes with a brief summary and outlook.

\section{Framework}
\setcounter{equation}{0}

We consider the representation of the resolvent algebra on Fock space
\mbox{$\cF \subset \cH$}, \ie the totally symmetric subspace of the 
unsymmetrized space 
$\cH = \oplus_{n = 0}^\infty  \, \cH_n$. 
In more detail, the normalized 
vacuum vector corresponding to $n=0$ is denoted by 
$\Omega$, the single particle space is 
$\cH_1 \simeq L^2(\RR^s)$
and the unsymmetrized $n$-particle space is the 
$n$-fold tensor product of the single particle space, 
$$
\cH_n \simeq 
L^2(\RR^s) \otimes \cdots \otimes L^2(\RR^s) \, , \quad  
n \in \NN \, .
$$ 
On $\cH_n$ acts the unitary representation $U_n$
of the symmetric group $\bSigma_n$, which for 
$\Psi_1, \dots , \Psi_n \in \cH_1$  is given by 
$$ 
U_n(\pi) \, \Psi_1 \otimes \cdots \otimes \Psi_n
=  \Psi_{\pi(1)} \otimes \cdots \otimes \Psi_{\pi(n)} \, , 
\quad \pi \in \bSigma_n \, .
$$
Its mean 
\ $\overline{U}_n(\Sigma_n) \doteq (1/n!) \, \sum_{\pi \in \sbSigma_n} U_n(\pi)$ 
\ over the group $\Sigma_n$ is the orthogonal projection
in $\cH_n$ onto the totally symmetric subspace $\cF_n$.
We define the $n$-fold symmetric tensor product 
of vectors  $\Psi_1, \dots , \Psi_n \in \cH_1$, putting
$$
\Psi_1 \otimes_s \cdots \otimes_s \Psi_n \doteq
(1/n!) \, \sum_{\pi \in \sbSigma_n} 
\Psi_{\pi(1)} \otimes \cdots \otimes \Psi_{\pi(n)} 
\in \cF_n \, .
$$
Similarly, we define the symmetric (symmetrized) tensor product of operators
$A_1, \dots , A_n$, acting on $\cH_1 = \cF_1$, by
$$
A_1 \otimes_s \cdots \otimes_s A_n \doteq
(1/n!) \sum_{\pi \in \sbSigma_n} A_{\pi(1)} \otimes \cdots \otimes A_{\pi(n)} \, .
$$
The symmetric subspace $\cF_n \subset \cH_n$ is stable under the 
action of these operators, $n \in \NN$. 

\medskip 
\noindent \textit{Fields:} \
On Fock space $\cF$ there act 
the creation and annihilation operators 
$a^*$ and $a$, which are regularized 
with complex-valued test functions $f,g \in \cD(\RR^s)$,
having compact support. 
They satisfy on their standard domain 
of definition the commutation relations given by 
$$
[a(f), a^*(g)] = \langle f, g \rangle \, \one \, ,
\quad [a(f), a(g)] = [a^*(f), a^*(g)] = 0 \, ,
$$
where \ $\langle f, g \rangle = \int \! d\bix \, \overline{f}(\bix) g(\bix)$. 
We recall that $a^*(f)$ is complex linear in $f$ whereas~$a(f)$, being
the hermitean conjugate of $a^*(f)$, is antilinear in $f$.
The operators~$a(f)$ annihilate the vacuum vector $\Omega$ 
and, for given $n \in \NN$ and $f_1, \dots, f_n \in \cD(\RR^s)$,  
the corresponding vectors in $\cF_n$ are created from $\Omega$ 
according to  
$$
|f_1 \rangle \otimes_s \cdots \otimes_s  | f_n \rangle =
(1 / n!)^{1/2} \ a^*(f_1) \cdots a^*(f_n) \, \Omega \, . 
$$

This algebraic structure can be rephrased in terms of a single real linear,  
symmetric field operator~$\phi$ given by 
$\phi(f) \doteq (a^*(f) + a(f))$ with $f \in \cD(\RR^s)$. 
It satisfies the commutation relations
$$
[\phi(f) , \phi(g)] = i \sigma(f,g) \, \one
\, , \quad f,g \in \cD(\RR^s) \, ,
$$
where $\sigma(f,g) \doteq i (\langle g, f \rangle - \langle f,g \rangle )$
is a non-degenerate real linear symplectic form on $\cD(\RR^s)$ 
(being regarded as a real symplectic space). One can recover the 
creation and annihilation operators from the field operator through
the relations
$$
2 a^*(f) = \phi(f) - i \phi(i f) \, , \quad 
2 a(f) = \phi(f) + i \phi(i f) \, , \quad f \in \cD(\RR^s) \, .
$$
The resolvents of the field operators,  
$$
R(\lambda,f) \doteq (i \lambda  + \phi(f))^{-1} \, , \quad 
\lambda  \in \RR \backslash \{ 0 \} \, , f \in \cD(\RR^s) \, , 
$$
generate, by taking their sums and products and
proceeding to the norm closure, the resolvent 
algebra $\bfR$, based on the symplectic space 
$(\cD(\RR^s), \sigma)$. As already mentioned, the 
algebra $\bfR$ provides a concrete and faithful
representation of the abstractly defined resolvent 
algebra \cite[Thm.~4.10]{BuGr2}. We therefore
fix it throughout the subsequent analysis.

On Fock space there acts the particle number operator $N$.
It is the generator of the group $\Gamma \simeq U(1)$ of gauge
transformations given by  
$$
e^{i s N} \phi(f) e^{-is N} = \phi( e^{is}  f) \, , \quad
s \in [0, 2 \pi] \, ,  \, f \in \cD(\RR^s) \, .
$$
We denote by $\bfA \subset \bfR$ the C*-subalgebra of operators,  
which are invariant under these gauge transformations. It contains
for example the resolvents of the operators 
$$
\big( \phi(f)^2 + \phi(if)^2 \big) = 4 \, a^*(f)a(f) + 2 \langle f, f \rangle 
\, \one \, , \quad f \in \cD(\RR^s) \, ,
$$ 
which are gauge invariant 
elements of the resolvent algebra $\bfR$. This fact can 
be inferred form  \cite[Prop.~4.1]{BuGr3},  
taking into account that $\sigma(f, i f) \neq 0$ if $f \neq 0$, whence 
$\phi(f)$ and $\phi(if)$ are canonically conjugate 
operators. 

\medskip 
\noindent \textit{Position and momentum:}  \ 
It will be convenient in our analysis of dynamics to deal
also with the quantum mechanical position and 
momentum operators. These operators are denoted 
by $\biQ$, $\biP$ and satisfy canonical 
commutation relations which, in an obious notation, are  
$$
[\bia \biQ, \, \bib \biP ] = i \, ( \bia \bib )
\, \one \, , \quad 
\bia, \bib \in \RR^s \, ,
$$
all other commutators being $0$. We make use of the 
Schr\"odinger representation of these operators, 
$$
\bix \mapsto (\bia  \biQ f)(\bix) = (\bia \bix)  f(\bix) \, , \quad
\bix \mapsto (\bib \biP f)(\bix) = -i (\bib \, \bpartial) f(\bix) \, ,
$$
where $f \in L^2(\RR^s)$ lies in their respective domains.
Given $n \in \NN$, we consider pairs of these operators,  
$\biQ_1, \biP_1$, \dots , $\biQ_n, \biP_n$, 
which act on the correspondingly numbered components of the 
unsymmetric $n$-fold tensor product 
\mbox{$\cH_n \simeq L^2(\RR^s) \otimes \cdots \otimes L^2(\RR^s)$}  
and commute amongst each other.
Symmetrized functions of these operators leave the 
Fock space $\cF_n \subset \cH_n$ invariant, a prominent
example being the restriction of the Hamiltonian in equation 
\eqref{e1.1} to $\cF_n$, 
\begin{equation} \label{e2.1}
H_n \doteq H \upharpoonright \cF_n 
= \sum_i \biP_i^2 + \sum_{j \neq k} V(\biQ_j - \biQ_k) \, .
\end{equation}
The sums involved here extend over $i,j,k \in \{1, \dots , n\}$.
As explained in \cite{BuGr2}, one can also define resolvent
algebras of position and momentum operators, but we make no use 
of this formalism here.

\section{Structure of observables}
\setcounter{equation}{0}

In this section we will clarify the structure of 
the gauge invariant subalgebra $\bfA$ of 
the resolvent algebra of Bose fields. In a first step we determine the  
properties of special elements of this algebra.

\begin{lemma} \label{l3.1}
Let $M = \prod_{k=1}^m R(\lambda_k, f_k) \in \bfR$
be any ordered product (monomial) of resolvents of the field,
$\lambda_k \in \RR \backslash \{ 0 \}$, 
$f_k \in \cD(\RR^s) \backslash \{ 0 \}$, $k = 1, \dots , m$. Its mean over the 
gauge group,   
$\overline{M}^\Gamma \doteq (2 \pi)^{-1} \int_0^{2 \pi} \! dt \, 
e^{itN} M e^{-itN}$, which is defined in the strong operator topology,  
is an element of $\bfA$. 
Moreover, denoting by~$L$ the complex linear span  
generated by $(f_1, \dots , f_m)$ and by
$\cF(L) \subset \cF$ the Fock space based on the subspace 
$L \subset L^2(\RR^s)$, the mean 
$\overline{M}^\Gamma$ acts on $\cF(L)$ as a compact operator. 
\end{lemma}

\begin{proof}
Noticing that the space $(L, \sigma)$ is a finite dimensional non-degenerate
symplectic subspace of $(\cD(\RR^s),\sigma)$, let $\fR(L) \subset \bfR$ 
be the resolvent algebra generated by the resolvents 
$R(\lambda,h)$, where $\lambda \in \RR \backslash \{ 0 \}$, $h \in L$.
This algebra acts irreducibly on 
$\cF(L)$ and one has $e^{isN} M e^{-isN} \in \fR(L)$, $s \in [0, 2 \pi]$. 
Consider now the function
$$
s,t \mapsto e^{isN} M^* e^{-isN} e^{itN} M e^{-itN} \, , \quad
s,t \in [0, 2 \pi] \, .
$$ 
Since 
$e^{itN} R(\lambda, h) e^{-itN} = R(\lambda, e^{it} h)$,
and similarly for the adjoints, 
the values of this function lie in the intersection of the 
principal ideals in 
$\fR(L)$, which are generated
by the individual gauge-transformed resolvents in the 
above product. According
to \cite[Prop.~4.4]{Bu2}, this intersection coincides 
with the principal ideal generated by the reordered product 
$$
R(\lambda_1, e^{is} f_1)^*  R(\lambda_1, e^{it} f_1)
\cdots R(\lambda_m, e^{is} f_m)^*  R(\lambda_m, e^{it} f_m) \, .
$$
According to~\cite[Thm.~5.4]{BuGr2}  
the latter operator acts as a compact operator on $\cF(L)$
if all adjacent pairs of resolvents are generated by canonically 
conjugate operators, \ie if 
$$
\sigma(e^{is} f_k, e^{it} f_k) = 
i (e^{i(t-s)} - e^{i(s-t)}) \, \langle f_k, f_k \rangle  \neq 0 
\quad \text{for} \quad k = 1, \dots, m \, .
$$ 
So the above function
has, for almost all $(s,t) \in [0, 2 \pi] \times [0, 2 \pi]$, 
values in compact operators on $\cF_L$ and it is bounded. 
Hence the double integral (defined in the strong operator 
topology on $\cF_L$) 
$$
\overline{M}^{\Gamma *} \, \overline{M}^\Gamma = 
\int_0^{2 \pi} \! \! ds \int_0^{2 \pi} \! \! dt \, 
e^{isN} M^* e^{-isN} e^{itN} M e^{-itN}
$$
is a compact operator as well. Taking its square root and performing a polar 
decomposition, we find that 
$\overline{M}^\Gamma \upharpoonright \cF(L)$ is compact and consequently
an element of the compact ideal of 
$\fR(L) \upharpoonright \cF(L)$, 
cf.~\cite[Thm.~5.4]{BuGr2}. Since $\fR(L)$
is faithfully represented on $\cF(L)$, we conclude that 
$\overline{M}^\Gamma \in \fR(L) \subset \bfR$. 
Moreover,~$\overline{M}^\Gamma$ commutes
by construction with the gauge transformations which 
shows that $\overline{M}^\Gamma \in \bfA$. \qed 
\end{proof} 

Every element of $\bfR$ and hence \textit{a fortiori} of $\bfA$
can be approximated in norm by sums of monomials of resolvents
and the unit operator. It therefore follows from Lemma~\ref{l3.1} 
that the finite sums of operators of the form $\sum_i c_i \overline{M_i}^\Gamma$
with $\overline{M_i}^\Gamma$ as in the lemma,  
including the unit operator $\one$, are norm dense in~$\bfA$. 
In fact, given any $A \in \bfA$ and $\varepsilon > 0$ there exists a 
sum of monomials $\sum_i c_i M_i$ such that 
$\| \sum_i c_i M_i - A \| < \varepsilon$. Since $A$ is gauge invariant,  
we obtain, taking a mean over the gauge group, \ 
$\| \sum_i c_i \overline{M_i}^\Gamma - A \| < \varepsilon$, as claimed.
Thus we conclude that the algebra $\bfA$ is generated by the 
unit operator and the gauge invariant operators in the minimal compact 
ideals of all subalgebras of $\fR(L) \subset \bfR$ corresponding 
to the finite dimensional non-degenerate symplectic subspaces 
$(L, \sigma) \subset (\cD(\RR^s), \sigma)$. 

In the next step we analyze in detail the restrictions
of the algebra $\bfA$ to the subspaces $\cF_n \subset \cF$ 
for fixed  $n \in \NN$. 
From the algebraic point of view, these restrictions define  
irreducible 
representations of $\bfA$ which are disjoint for different values of 
$n$. In order to describe their structure, we introduce
the following quantities: given $1 \leq m \leq n$, 
let $\fC_m$ be the C*-algebra of compact operators on $\cF_m$.
This algebra coincides with the unique $m$-fold symmetric 
(symmetrized) tensor product
of the algebra of compact operators on~$\cF_1$, 
$$
\fC_m = \underbrace{\fC_1 \otimes_s \cdots \otimes_s \fC_1}_m \, .
$$ 
For $m=0$ we put $\fC_0 = \CC \, 1$. 
We embed these algebras into the algebra of bounded operators 
on $ \cF_n $, putting
\begin{equation} \label{e3.1}
\fC_{m,n} \doteq  \, \fC_m \otimes_s 
\underbrace{1 \otimes_s \cdots \otimes_s 1}_{n-m}  
= \, \underbrace{\fC_1 \otimes_s \cdots \otimes_s \fC_1}_m \otimes_s 
\underbrace{1 \otimes_s \cdots \otimes_s 1}_{n-m} \, , \quad
m \leq n \, .
\end{equation}
In particular, $\fC_{0,n} = \CC \, \one_n$ and $\fC_{n,n} = \fC_n$.
The subspaces $\fC_{m,n} \subset \cB(\cF_n)$ 
are linearily independent for different 
$m \leq n$ and the embeddings are continuous,~\ie 
$$
\| C \otimes_s \underbrace{1 \otimes_s \cdots \otimes_s 1}_{n-m} \|_n \leq  
\| C \|_m \, , \quad C \in \fC_m \, , \ 0 \leq m \leq n \, ,
$$
where $\| \, \cdot \, \|_n$ denotes the operator norm on 
$\cF_n$, $n \in \NN_0$. 
The latter assertion follows from the
fact that the operators in $\fC_{m,n}$ are obtained by taking
a mean over the permutation group $\Sigma_n$ of the corresponding
unsymmetrized tensor products on~$\cH_n$.
With these ingredients we proceed to the following definition.

\medskip 
\noindent \textbf{Definition:} \ Let $n \in \NN_0$. The algebra 
$\fK_n$ is the unital C*-algebra on $\cF_n$ that is generated by 
$\fC_{m,n}$, $0 \leq m \leq n$. It coincides with the
span of $\fC_{m,n}$, $0 \leq m \leq n$. 
The canonical (positive, linear) embeddings 
of these algebras into each other are 
denoted by $\epsilon_n: \fK_n \rightarrow \fK_{n+1}$. They are defined by 
$\epsilon_n(K_n) \doteq K_n \otimes_s 1$
and satisfy $\| \epsilon_n(K_n)\|_{n+1} \leq \| K_n \|_n$, $K_n \in \fK_n$.
The family $(\fK_n, \epsilon_n)_{n \in \NN_0}$ constitutes a 
directed system of C*-algebras. 

\medskip 
Note that the algebras  $\fK_n$ are AF-algebras, \ie
approximately finite dimensional, since the algebra of compact
operators is of this type. 
After these preparations we can establish the following fact, 
making use of the preceding lemma.

\begin{lemma} \label{l3.2}
Let $n \in \NN_0$. Then $\bfA \upharpoonright \cF_n \subseteq \fK_n$. 
\end{lemma}
\begin{proof}
As already explained, the elements of $\bfA$ can be approximated in the 
norm topology by sums of gauge-averaged monomials $\overline{M}^\Gamma$.
So it suffices to show that 
$\overline{M}^\Gamma \upharpoonright \cF_n \in \fK_n$ for any such 
operator. 

Let $L \subset L^2(\RR^s)$ be the 
subspace generated by all test functions appearing 
in the resolvents, which are factors of the given monomial 
$M$, cf.~Lemma~\ref{l3.1}. This space determines the resolvent algebra 
$\fR(L)$, the Fock space $\cF(L)$ and the particle number 
operator $N(L)$ acting on it. Since $L$ is finite dimensional, 
the resolvent of $N(L)$ is a compact operator on $\cF(L)$ and hence 
belongs to the compact ideal of $\fR(L)$, cf.~\cite[Thm.~5.4]{BuGr2}.
As a matter of fact, since this resolvent is also gauge invariant, it 
belongs to $\bfA(L)$. 
The same is true for its finite dimensional spectral
projections $E_m(L)$, commuting with $\overline{M}^\Gamma$,
which are attached to the
spectral values $m \in \NN_0$. We decompose 
$\overline{M}^\Gamma$ into operators of finite rank,
$\overline{M}^\Gamma = \sum_{m=0}^\infty \overline{M}^\Gamma E_m(L)$.
The sum converges in norm in view of the compactness
of $\overline{M}^\Gamma$ on~$\cF(L)$, proved 
in the preceding lemma. 

We can determine now the action of the operators 
$\overline{M}^\Gamma E_m(L)$, $m \in \NN_0$, on $\cF_n$.
To this end we decompose $\cF_n$ into the orthogonal sum 
of symmetric tensor products 
$$ 
\cF_n = \sum_{k = 0}^n \underbrace{L \otimes_s \cdots \otimes_s L}_k \otimes_s 
\underbrace{L^\perp \otimes_s \cdots \otimes_s L^\perp }_{n-k} \, ,
$$
where $L^\perp$ denotes the orthogonal complement of $L$ in 
$L^2(\RR^s)$. Since the elements of $\bfA(L)$ commute with the 
creation operators 
$a^*(h^\perp)$, $h^\perp \in L^\perp$, 
the operators in~$\bfA(L)$ act non-trivially 
(\ie differ from a multiple of the identity) only 
on factors in the tensor products which are contained in $L$. Thus,
taking also into account the 
action of the projections $E_m(L) \in \bfA(L)$, $m \in \NN_0$, 
on the subspaces 
$\cF_k(L) \subset \cF(L)$, $k \in \NN_0$, we obtain for $m \leq n$  
$$
\overline{M}^\Gamma E_m(L) \upharpoonright \cF_n  =
\big( \overline{M}^\Gamma E_m(L) \upharpoonright 
\underbrace{L \otimes_s \cdots \otimes_s L}_m \big) \otimes_s 
\underbrace{L^\perp \otimes_s \cdots \otimes_s L^\perp }_{n-m}  \, ,
$$
where  
$\, \overline{M}^\Gamma E_{0}(L) = 
\langle \Omega, \overline{M}^\Gamma \Omega \rangle \, E_{0}(L)$. 
If $m > n$ we have   
$\overline{M}^\Gamma E_m(L) \upharpoonright \cF_n = 0$. 

Now, given $m \in \NN_0$, one clearly has 
$E_m(L) \cF_m \subset \cF_m(L)$. Hence $\overline{M}^\Gamma E_m(L)$
acts as a compact (even finite rank) operator $C_m$ on $\cF_m$
according to Lemma~\ref{l3.1}. 
Denoting by $F_L$ the finite dimensional projection on
$L^2(\RR^s)$ onto 
$L \subset L^2(\RR^s)$, the preceding step therefore implies that 
for any $m \leq n$
$$
\overline{M}^\Gamma E_m(L) \upharpoonright \cF_n 
= (C_m \underbrace{F_L \otimes_s \cdots \otimes_s F_L}_m \, ) \otimes_s 
\underbrace{(1 -F_L) \otimes_s \cdots \otimes_s (1 - F_L)}_{n-m}  \in
\fK_n \, .
$$ 
Since $\overline{M}^\Gamma \upharpoonright \cF_n 
= \big( \sum_{m = 0}^n \, \overline{M}^\Gamma E_m(L) \big) \upharpoonright \cF_n $,
this completes the proof of the statement. \qed 
\end{proof}

In the next step we show that $\bfA \upharpoonright \cF_n$ 
actually coincides with $\fK_n$. 
Before we turn to the proof, let us point to the following facts:
wheras the algebra~$\bfA$ is faithfully represented on $\cF$
and acts irreducibly on each $\cF_n$, $n \in \NN_0$, its action 
on these subspaces is not faithful. We denote 
by $\bfJ_{n} \subset \bfA$ the closed two-sided ideal of operators 
that are annihilated on $\cF_n$ and show below that these ideals 
are nested, \ie $\bfJ_{n+1} \subset \bfJ_{n}$. Hence the  
seminorms $\| \, \cdot \, \|_n$ on $\bfA$
are increasing, $n \in \NN_0$, and \
$\| \, \cdot \, \| \doteq \sup_{n} \| \, \cdot \, \|_n$ 
agrees with the C*-norm on $\bfA$.

\begin{lemma} \label{l3.3}
Let $n \in \NN_0$. Then $\bfA \upharpoonright \cF_n = \fK_n$. 
\end{lemma}

\begin{proof}
Since the quotient C*-algebra 
$\bfA / \bfJ_{n}$ is faithfully represented on $\cF_n$, it suffices
for the proof of the statement to show   
that $\bfA \upharpoonright \cF_n$ is dense in $\fK_n$ with 
regard to the operator norm on $\cF_n$. We will accomplish this by showing
that for each $0 \leq m \leq n$ the algebra $\bfA \upharpoonright \cF_n$
includes all operators in $\fC_{m,n} \subset \fK_n$ that arise from operators  
in $\fC_m$ having finite rank. Since the embedding of $\fC_m$ 
into $\fC_{m,n}$ is norm continuous, the statement then follows. 

For $m = 0$ we have $\fC_{0, n} = \CC \, \one_n$, so there is nothing
to prove. Next, let $m = 1$ and let $\{e_j \in \cD(\RR^s) \}_{j \in \NN}$ be  
an orthonormal basis in $L^2(\RR^s) \simeq \cF_1$. The 
algebra $\fC_1$ contains in particular the matrix units corresponding 
to this basis and their linear span is norm dense in $\fC_1$. 
The action of these matrix units on the vectors $\Phi_1 \in \cF_1$ is given by
$$
M_{i k} \Phi_1 = \langle e_k, \Phi_1 \rangle \, e_i \, , \quad
i, k \in \NN \, .
$$
We recall that their embeddings into $\fC_{1,n}$ are defined by the
symmetric (symmetrized) tensor product 
$M_{i k} \otimes_s \underbrace{1 \otimes_s \dots \otimes_s 1}_{n-1}$, \
$i,k \in \NN$. 

In order to establish the existence of
operators in $\bfA$ that induce the same action
on $\cF_n$ for given $n \in \NN$, 
we consider the complex rays $L_j \doteq \, \CC \, e_j$ and the 
corresponding particle number operators $N(L_j) = a^*(e_j) a(e_j)$ 
on the Fock spaces $\cF(L_j)$, $j \in \NN$. 
As already explained, all spectral projections of $N(L_j)$
corresponding to spectral values in the finite interval $[0, n]$
are finite dimensional and thus are 
elements of the compact ideal of the resolvent algebra $\fR(L_j)$.
Denoting the sum of these projections by $E_{[0,n]}(L_j)$, we 
therefore have 
$$
X_n(L_j) \doteq 
n^{-1/2} \, E_{[0,n]}(L_j) \, a(e_j) \, \in \, 
\fR(L_j) \subset \fR \, , \quad 
j \in \NN_0 \, .
$$
The resulting 
operators $W_n(i,k) \doteq X_n(L_i)^* X_n(L_k)$ are gauge invariant, hence 
elements of $\bfA$, and a routine computation shows that 
$$
W_n(i,k) \upharpoonright \cF_n = 
M_{i k} \otimes_s \underbrace{1 \otimes_s \dots \otimes_s 1}_{n-1} \, ,
\quad i,k \in \NN \, .
$$
Since all operators in $\fC_1$ of finite rank can be decomposed into 
finite sums of matrix units, this proves the statement for $m = 1$.

For fixed $n \in \NN$, the operators $W_n(i,k)$ can now be used in order to 
obtain operators of finite rank in the spaces $\fC_{m, n}$ for any
$m \leq n$. One proceeds from the equality
\begin{align*}
W_n(i_1, k_1) & \cdots W_n(i_m, k_m) \upharpoonright 
\cF_n \\
= & \, (M_{i_1 k_1} \otimes_s \underbrace{1 \otimes_s \dots \otimes_s 1}_{n-1})
\cdots (M_{i_m k_m} \otimes_s \underbrace{1 \otimes_s \dots \otimes_s 1}_{n-1}) \, .
\end{align*}
Due to the symmetrization and the relation 
$M_{i^\prime k^\prime} M_{i^{\prime \prime} k^{\prime \prime}} = 
\delta_{k^\prime i^{\prime \prime}} M_{i^\prime k^{\prime \prime}}$, 
the right hand side of this 
equality is a linear combination of operators in 
$\fC_{m^\prime, n}$ which arise from matrix units 
in $\fC_{m^\prime}$, $1 \leq   m^\prime \leq m$. By an obvious recursive  
procedure one can determine certain specific linear combinations of 
these operators where all contributions involving products of
the matrix units are eliminated and only their $m$-fold tensor 
product remains. The result is of the form
$$
W_n(i_1, \dots i_m, k_1, \dots k_m) \! \doteq \! \sum c_{j_1, \dots j_{m^\prime}, 
l_1, \dots l_{m^\prime}}
W_n(j_1,l_1) \! \cdots \! W_n(j_{m^\prime}, l_{m^\prime}) \, ,
$$ 
where the sum extends over  
$j_1, \dots , j_{m^\prime} \in \{ i_1, \dots i_m  \}$, \ 
$l_1, \dots , l_{m^\prime} \in \{ k_1, \dots k_m  \}$
and $m^\prime \leq m$. 
With properly determined coefficents, these operators satisfy   
$$
W_n(i_1, \dots i_m, k_1, \dots k_m) 
\upharpoonright \cF_n
= M_{i_1 k_1} \otimes_s \cdots \otimes_s M_{i_m k_m}
\otimes_s \underbrace{1 \otimes_s \cdots \otimes_s 1}_{n-m} \, .
$$
Hence $\bfA \upharpoonright \cF_n$ contains
for any given $0 \leq m \leq n$ 
all operators in $\fC_{m, n}$ that arise from operators 
in $\fC_m$ of finite rank. Recalling that these operators are 
norm dense in the compact operators, we conclude that 
$\bfA \upharpoonright \cF_n$
is norm dense in $\fK_n$ with regard to the operator norm
on $\cF_n$. Bearing in mind the initial remarks, this completes 
the proof of the statement.  \qed \end{proof}

The preceding lemma shows that the restrictions of the
elements $A \in \bfA$ to the $n$-particle subspaces of Fock space 
correspond to elements of the directed system 
$(\fK_n, \epsilon_n)_{n \in \NN_0}$.
These restrictions are representations
$\rho_n$ of the algebra $\bfA$, $n \in \NN_0$, 
$$
\rho_n(A) \doteq A \upharpoonright \cF_n \, , \quad A \in \bfA \, .
$$
Lemma \ref{l3.3} implies that for any $A \in \bfA$ one has
$\rho_n(A) = \sum_{m=0}^n C_{m,n} \in \fK_n$ for elements 
$C_{m,n} \in \fC_{m,n}$, $0 \leq m \leq n$. Conversely, given 
any element of $\fK_n$, there is some operator 
$A \in \bfA$ satisfying this equality.
In order to better understand the global structure of
this correspondence, we make use of the kinematical 
locality properties of $\bfA$ and clustering properties of the 
states in $\cF$. The results will put us into the position  
to introduce canonical maps from 
$\fK_n$ to $\fK_{n-1}$,  $n \in  \NN$. So these maps go into 
the inverse direction of the above directed system. 

Given $n \in \NN$, we consider vectors in $\cF_n$,    
where one of the single particle components undergoes large spatial
translations. Picking $f_1, \dots , f_n \in \cD(\RR^s)$,  we define 
$$
\bPhi^n(\bix) \doteq | f \rangle_1 \otimes_s \dots \otimes_s | f_{n-1} \rangle 
\otimes_s | f_n(\bix) \rangle \, \in \cF_n \, ,
$$
where $f_n(\bix)$ denotes the translated function $f_n$, hence 
$| f_n(\bix) \rangle = e^{i \sbix \sbiP} \, | f_n \rangle$, 
\mbox{$\bix \in \RR^s$}. We also define 
$$
\bPhi^{n-1} \doteq | f_1 \rangle \otimes_s \dots \otimes_s |f _{n-1} \rangle 
\in \cF_{n-1} \, ,
$$
where we put $\bPhi^0 \doteq \Omega$. Replacing the functions 
$f_k$ by $g_k \in \cD(\RR^s)$, $k = 1, \dots , n$,  we 
get in a similar manner vectors
$\Psi^n(\bix) \in \cF_n$ and $\Psi^{n-1} \in \cF_{n-1}$.
With this notation, the following lemma obtains.

\begin{lemma} \label{l3.4}
Let $A \in \bfA$ and let $n \in \NN$. Then 
$$ \lim_{\sbix \rightarrow \infty}  \,
\langle \bPsi^n(\bix), \, 
A \ \bPhi^n(\bix) \rangle = n^{-1} \langle \bPsi^{n-1}, \, 
A \, \bPhi^{n-1} \rangle \
\langle g_n , f_n \rangle \, .
$$ 
As a consequence, \ $\| A \|_{n-1} 
\leq \| A \|_n$.
Recalling that $A \upharpoonright \cF_n 
= \sum_{m = 0}^n C_{m,n}$, where $C_{m,n} \in \fC_{m,n}$, 
and omitting from each operator $C_{m,n}$ 
a tensor factor~$1$, one obtains operators  
$C_{m,n-1} \in \fC_{m,n-1}$, \ $0 \leq m \leq n-1$,
and 
$$
A \upharpoonright \cF_{n-1} = 
\sum_{m = 0}^{n-1} \, (n-m)/n \ C_{m, n-1} \, .  \hspace{10mm} \qed
$$
\end{lemma}

\noindent \textbf{Remark:} \ It follows from this 
lemma that every element $A \in \bfA$ which annihilates all vectors
in~$\cF_n$ also annihilates all vectors in $\cF_{n-1}$.
This implies the inclusion of ideals  
\mbox{$\fJ_n \subset \fJ_{n-1}$} mentioned above, $n \in \NN$. 

\begin{proof}
The first equality is a well known consequence of the canonical 
commutation relations. In the case at hand, these relations imply 
$$
[R(\lambda, f), a^*(f_n(\bix))] = 
\langle f, f_n(\bix) \rangle \, R(\lambda, f)^2 \, ,
\quad f \in \cD(\RR^s) \, .
$$
Now $|f_n(\bix) \rangle = e^{i \sbix \sbiP} | f_n \rangle$ and 
$e^{i \sbix \sbiP} \rightarrow 0$  weakly on $L^2(\RR^s)$
as $\bix \rightarrow \infty$; hence 
$\langle f, f_n(\bix) \rangle \rightarrow 0$. One can therefore commute 
in this limit the creation operator~$a^*(f_n(\bix))$, 
used in creating the component $| f_n(\bix) \rangle$ of
$\Phi^n(\bix)$ from the vacuum, to the left of 
any given monomial $M$ in the resolvents. Its adjoint then acts on the 
vector $ \bPsi^n(\bix)$. Commuting it through the
corresponding creation operators to the vacuum, the statement 
follows in this special case since 
$\langle g_k, f_n(\bix) \rangle \rightarrow 0$, $k = 1, \dots , n-1$, 
and $\langle g_n(\bix), f_n(\bix) \rangle 
= \langle g_n, f_n \rangle$.
Since the finite sums of monomials are norm-dense in~$\bfA$,
this result extends to all $A \in \bfA$. 

The inequality involving the norms of $A$ follows from this result 
by fixing normalized functions $g_n = f_n$. Noticing that
the stated equality then holds for arbitrary linear combinations 
$\bPhi^{n-1}, \bPsi^{n-1} \in \cF_{n-1}$ 
of the respective factors in the vectors
$\bPhi^n(\bix), \bPsi^n(\bix) \in \cF_n$, one obtains 
\begin{align*}
& |\langle \bPsi^{n-1}, A \, \bPhi^{n-1} \rangle| \, \big/ \, 
\| \bPsi^{n-1} \| \, \| \bPhi^{n-1} \| \\
& = \lim_{\sbix \rightarrow \infty} 
| \langle \bPsi^n(\bix), A \, \bPhi^n(\bix) \rangle | \, \big/ \, 
\| \bPsi^n(\bix) \| \, \| \bPhi^n (\bix) \| \ \leq \ \| A \|_n \, .
\end{align*}
This implies $\| A \|_{n-1} \leq \| A \|_n$.

Finally, let 
$A \upharpoonright \cF_n = \sum_{m = 0}^n C_{m,n}$. 
Replacing $A$ by this sum, one
obtains matrix elements of the form 
$ \langle \bPsi^n(\bix), \, C_{m,n} \ \bPhi^n(\bix) \rangle$, \
$0 \leq m \leq n$. Since the algebra $\fC_m$ is the 
$m$-fold symmetric C*-tensor product of the algebra
of compact operators $\fC_1$ on $\cF_1$, it suffices to 
consider operators $C_{m,n}$ which are the $n$-fold symmetric
tensor product of $m$ compact operators 
$C^{(1)}, \dots , C^{(m)}$ on $\cF_1$ and $n-m$ unit 
operators. Now if one of the operators 
$C^{(k)}$, $k = 1, \dots, m$, acts on the 
component $| f_n(\bix) \rangle$ of
$\bPhi^n(\bix)$ it maps, being compact, these weakly 
convergent vectors into vectors which converge strongly 
to $0$ for $\bix \rightarrow \infty$. An 
analogous statement holds for the adjoints of 
$C^{(k)}$, $k = 1, \dots, m$, acting on~$\bPsi^n(\bix)$. 

In order to see which terms in the above matrix element 
do not vanish 
in the limit $\bix \rightarrow \infty$, we consider on $\cH_n$ the 
unymmetrized tensor product 
$$
C^{(1)} \otimes \cdots \otimes C^{(m)}
\otimes \underbrace{1 \otimes \cdots \otimes 1}_{n-m} 
$$
between the symmetrized vectors. The subsequent symmetrization of this 
operator therefore does not change the result. 
Let $\pi, \pi^\prime \in \Sigma_n$ be the permutations
used in the symmetrization of $\bPhi^n(\bix)$, respectively
$\bPsi^n(\bix)$. According to the preceding remarks, 
the terms surviving in the limit arise for 
permutations $\pi$ and $\pi^\prime$
satisfying $\pi(n) = \pi^\prime(n)  \in \{m+1, \dots, n \}$.  
If $m = n$ there are no such permutations, \ie the corresponding
matrix element vanishes in the limit. If $0 \leq m < n$ one
obtains for each of the 
$n - m$ admissible values $\pi(n) = \pi^\prime(n)$ a total of 
$(n - 1)!$ permutations $\pi$, respectively $\pi^\prime$, of the
remaining indices which satisfy 
this condition. This gives
$$
\lim_{\sbix \rightarrow \infty} 
\langle \bPsi^n(\bix), \, C_{m,n} \ \bPhi^n(\bix) \rangle 
= (n-m)/n^2 \, \langle  \bPsi^{n-1}, \, C_{m,n-1} \, \bPhi^{n-1} \rangle \ 
\langle g_n, f_n \rangle \, .
$$
The formula for the restriction of $A$ to $\cF_{n-1}$ now follows 
by taking the sum of these terms for $m = 0, \dots ,n$ and 
comparing the resulting equality with the one obtained 
for $A$ in the first step.
\qed \end{proof}

The preceding lemma establishes a relation between the representations
$\rho_n$ and $\rho_{n-1}$ of $\bfA$ and hence for all representations with 
smaller particle number. Based on the information obtained so far, we 
define now inverse maps $\kappa_n$ between the corresponding 
algebras $\fK_n$ and $\fK_{n-1}$, $n \in \NN$.\footnote{The maps $\kappa_n$ 
go into the opposite direction of the maps 
$\epsilon_n$, defined above, but they are not 
their (left) inverses. This could be rectified by linear 
transformations on $\fK_n$.} 

\medskip 
\noindent \textbf{Definition:} Let $n \in \NN_0$. The (surjective) map 
$\kappa_n : \fK_n \rightarrow \fK_{n-1}$ is defined by 
$$ 
\kappa_n\big(\sum_{m = 0}^n C_{m,n} \big)
\doteq \sum_{m = 0}^{n-1} (n-m)/n \ C_{m,n-1} \, ,
$$
where $C_{m,n} \in \fC_{m, n}$,  and 
$C_{m,n-1} \in \fC_{m, n-1}$ is obtained from
$C_{m,n}$ by omitting a tensor factor $1$, $0 \leq m \leq n-1$. 
For $n=0$ we put $\kappa_0 = 0$. The Banach space of bounded sequences 
$\biK \doteq \{ K_n \in \fK_n \}_{n \in \NN_0}$
satisfying the coherence condition 
$\kappa_n(K_n) = K_{n-1}$, $n \in \NN_0$, 
is denoted by~$\bfK$. 
The corresponding norm is given by   
$\| \biK \|_\infty \doteq \sup_n \| K_n \|_n$, $ \biK \in \bfK$. 
By some abuse of the terminology used 
in \cite{Ph}, $\bfK$ is called the (bounded) inverse  
limit of the inverse system $( \fK_n, \kappa_n )_{n \in \NN_0}$.

\medskip 
The following picture emerges from these results.  
According to Lemma \ref{l3.3} one has $\rho_n(\bfA) = \fK_n$. 
Hence, by the above definition, $\kappa_n$ maps 
$\rho_n(\bfA)$ onto $\rho_{n-1}(\bfA)$. As a matter of fact, 
it follows from Lemma \ref{l3.4} that this map is a homomorphism
since $\kappa_n(\rho_n(A)) = \rho_{n-1}(A)$, $A \in \bfA$. 
Thus for $\rho_n(A_1), \rho_n(A_2) \in \rho_n(\bfA)$ one obtains  
\begin{align*}
& \kappa_n(\rho_n(A_1)) \,  \kappa_n(\rho_n(A_2))
= \rho_{n-1}(A_1) \, \rho_{n-1}(A_2) \\
& \hspace*{21.5mm} = \rho_{n-1}(A_1 A_2) 
= \kappa_n(\rho_n(A_1 A_2)) =   \kappa_n(\rho_n(A_1) \rho_n(A_2)) \, , \\
& \kappa_n(\rho_n(A_1))^* = \rho_{n-1}(A_1)^* = \rho_{n-1}(A_1^*) =
\kappa_n(\rho_n(A_1^*)) = \kappa_n(\rho_n(A_1)^*) \, , 
\end{align*}
proving the statement. Lemma \ref{l3.4} also implies that 
$\| \kappa_n(\rho_n(A)) \|_{n-1} 
\leq  \| \rho_n(A) \|_n$, $A \in \bfA$. Thus 
$\kappa_n$ acts continuously on $\fK_n$, 
$$
\| \kappa_n(K_n) \|_{n-1} \leq \| K_n \|_n \, , \quad K_n \in \fK_n \, .
$$  

It follows from these observations that the inverse limit 
$\bfK$ is a C*-algebra, where the  
algebraic operations are pointwise defined, \eg 
$$
\biK_1 \biK_2 \doteq \{ K_{1 n} K_{2 n} \}_{n \in \NN_0} \, , \quad
\biK_1^* \doteq \{ K_{1 n}^* \}_{n \in \NN_0} \, , \quad \biK_1, \biK_2 \in \bfK \, .
$$
Its C*-norm is given by
$\| \biK \|_\infty = \sup_n \| K_n \|_n$, $\biK \in \bfK$. 
Moreover, $\bfK$ is faithfully represented on Fock space: let  
$\Phi = \sum_{n=0}^\infty \Phi_n \in \cF$, where 
$\Phi_n \in \cF_n$ and $\sum_{n=0}^\infty \| \Phi_n \|^2 < \infty$. 
The representation of~$\bfK$, denoted by $\rho$, is obtained by putting  
$$
\rho(\biK) \, \Phi \doteq \sum_{n=0}^\infty 
K_n \Phi_n \, , \quad \biK \in \bfK \, .
$$
Note that 
$\| \rho(\biK) \| = \sup_n \| K_n \|_n = \| \biK \|_\infty$, so the
representation is faithful. 

\medskip 
The inverse limit $\bfK$ will prove to be a convenient tool in the analysis 
of the action of dynamics. It is therefore gratifying that the 
algebra~$\bfA$ can be extended in a straightforward manner to an algebra which 
matches with the inverse limit: according to Lemmas \ref{l3.2} 
and \ref{l3.4},  every element $A \in \bfA$ defines a sequence  
$\biK(A) = \{ \rho_n(A) \in \fK_n \}_{n \in \NN_0}$ 
which is coherent, $\kappa_n(\rho_n(A)) = \rho_{n-1}(A)$, $n \in \NN_0$. 
Hence, $\biK(A) \in \bfK$.
There is also a certain converse: given a coherent sequence 
$\biK = \{ K_n \in \fK_n \}_{n \in \NN_0} \in \bfK$,  there 
exists  according to Lemma~\ref{l3.3} for each $n \in \NN_0$
some operator $A_n \in \bfA$ such that 
$\rho_n(A_n) = K_n$. Lemma~\ref{l3.4} then implies that 
$K_{n-1} = \kappa_n(K_n) = \kappa_n(\rho_n(A_n)) = \rho_{n - 1}(A_n)$.
Iterating the inverse maps, one arrives at the relations 
$\rho_m(A_n) = K_m$, $m = 0, \dots , n$. 
So the operator $A_n$ reproduces the first \ $n+1$ terms of
the given coherent sequence. The coherent sequence 
$\{ \rho_n(A_n) = K_n\}_{n \in \NN_0}$, resulting from
this construction, therefore reproduces the given $\biK$.
Since each $\biK \in \bfK$ is represented by a bounded 
operator on $\cF$, this leads us to the following simple 
characterization  of the mentioned extension of $\bfA$. 

\medskip \noindent
\textbf{Definition:} The algebra $\obfA$, extending $\bfA$, consists of 
the family of bounded operators on $\cF$ whose members $A$ satisfy
$$
A \upharpoonright \oplus_{m = 0}^n \, \cF_m \in 
\bfA \upharpoonright \oplus_{m = 0}^n \, \cF_m \, , \quad 
n \in \NN_0 \, .
$$
Thus $\obfA$ and $\bfA$ differ as sets only on states with an infinite 
particle number.

\medskip \noindent
\textbf{Remark:} 
Simple examples of operators in $\obfA$ which are not 
contained in $\bfA$ are bounded functions of the 
particle number operators $N(L)$ for 
finite dimensional symplectic subspaces $L \subset \cD(\RR^s)$.
Since the spectral projections of $N(L)$ are contained
in~$\bfA$, the restriction of any such function of $N(L)$ to 
$\oplus_{m=0}^n \cF_m$ is represented by some
element of $\bfA$, $n \in \NN_0$.
But the operator itself is not contained 
in $\bfA$ if the underlying function does not tend to 
a constant asymptotically.

\bigskip
It is apparent from the preceding discussion that the
algebra $\obfA$ is isomorphic to the inverse limit $\bfK$. 
In particular, it is a C*-algebra with norm given by
$\| A \|_\infty = \sup_n \| \rho_n(A) \|_n$, $A \in \obfA$.
Even though the above alternative definition, 
relating~$\obfA$ directly to $\bfA$, is intuitively attractive, 
we will rely in the subsequent analysis on the more 
detailed information contained in  $\bfK$. 
We summarize the results obtained in this section in the following
theorem.

\begin{theorem} \label{t3.5} 
The map
$$
A \mapsto 
\{ \rho_n(A) = A \upharpoonright \cF_n \in \fK_n \}_{n \in \NN_0} \, ,
\quad A \in \overline{\bfA} \, ,
$$
establishes an isomorphism between the algebra $\overline{\bfA}$ 
and the inverse limit $\bfK$ of the inverse system of 
approximately finite dimensional 
C*-algebras $( \fK_n, \kappa_n )_{n \in \NN_0}$.  
\end{theorem}

\section{Dynamics of observables} \label{s4}
\setcounter{equation}{0}

With this information about the structure of the
algebra $\overline{\bfA}$, we can turn
now to the discussion of dynamics. We will take advantage of the 
fact that the elements of the
algebra $\overline{\bfA}$ and the Hamiltonians under 
consideration commute with the particle number operator. So 
for any $n \in \NN$ we can 
restrict these operators to the subspaces $\cF_n \subset \cF$. We recall that 
$\overline{\bfA} \upharpoonright \cF_n = \fK_n$, the algebra
generated by symmetrized tensor products of compact operators and the unit
operator. The restrictions of the 
Hamiltonian $H  \upharpoonright \cF_n = H_n$ can be expressed
in terms of position and momentum operators, cf.\ 
equation~\eqref{e2.1}. 

It will be convenient in our analysis to extend 
the restricted Hamiltonian $H_n$ to 
the unsymmetric space $\cH_n$ by assigning the
underlying position and momentum operators to the
respective single-particle tensor factor with the same number.
We also extend 
the symmetric algebra $\fK_n$ on $\cF_n$ and embed it into 
the algebra $\fK(\II_n)$, $\II_n = \{1, \dots , n\}$,    
generated by all unsymmetrized compact and unit operators acting on 
the tensor factors of $\cH_n$. 
In more detail, given any ordered subset 
$\II_m \doteq \{ i_1, \dots , i_m \} \subset \II_n $
of $m \leq n$ elements,
we consider on $\cH_n$ the unique C*-tensor product 
$$
\fC_n(\II_m) \doteq \fC(i_1) \otimes \cdots \otimes \fC(i_m) 
$$
generated by the commuting algebras of compact operators
acting on the corresponding tensor factors of $\cH_n$.
Here we have omitted tensor factors of 
$1$ acting on the remaining components. The unital C*-algebra 
$\fK(\II_n)$ is then defined as the linear span of the  
algebras $\fC_n(\II_m)$ for all $\II_m \subset \II_n$
and $0 \leq m \leq n$, where we put $\II_0 = \emptyset$ and 
$\fC_n(\II_0) \doteq \CC \, \one_n$. The symmetric algebra $\fK_n$ 
is identified with the subalgebra of $\fK(\II_n)$
consisting of all operators which commute with the unitaries  
$U_n(\pi), \pi \in \Sigma_n$. 

Making use of arguments established in \cite{Bu3} for distinguishable 
particles oscillating about lattice points, we will show that the 
algebra $\fK(\II_n) \subset \cB(\cH_n)$ is stable
under the adjoint action of the unitaries 
$e^{it H_n}$ on $\cB(\cH_n)$, determined by the 
Hamiltonian $H_n$. This action is denoted by 
$\alpha_n(t) \doteq \mbox{Ad} \, e^{it H_n}$, $t \in \RR$.
We recall here these arguments for the sake of completeness. 

For potential $V=0$ one obtains the  
non-interacting Hamiltonian $H_{0 n}$.
The resulting adjoint action $\alpha_n^{(0)}(t) \doteq \mbox{Ad} 
\, e^{itH_{0 n}}$ leaves the subalgebra 
$\fK(\II_n) \subset \cB(\cH_n)$ invariant, 
$t \in \RR$. This is apparent since it does not mix 
tensor factors and the adjoint action of a unitary operator maps     
compact operators onto compact operators.
Morover, since the function $t \mapsto e^{itH_{0 n}}$
is continuous in the strong operator topology, it is also
clear that $t \mapsto \alpha_n^{(0)}(t)$ acts 
pointwise norm-continuously on~$\fK(\II_n)$.

Next, we consider the Dyson operators  
$\Gamma_n(t) \doteq e^{it H_{0 n}} e^{-it H_n}$, $t \in \RR$. 
They define the Dyson maps 
$\gamma_n(t): \cB(\cH_n) \rightarrow \cB(\cH_n)$ given by 
\begin{equation} \label{e4.1}
\gamma_n(t) \doteq \mbox{Ad} \, \Gamma_n(t) \, , \quad 
t \in \RR \, .
\end{equation}
We must show that their restrictions 
to $\fK(\II_n) \subset \cB(\cH_n)$  
map this subalgebra onto itself. For this implies 
by the preceding remarks 
that the algebra is stable under the automorphic action 
$\alpha_n(t) = \alpha_n^{(0)}(t) \, \scirc \, \gamma_n(-t)$, $t \in \RR$, 
of the given dynamics. As a matter of fact, 
it suffices to establish the inclusion 
$\gamma_n(t)(\fK(\II_n)) \subset \fK(\II_n)$
since this implies 
$\alpha_n(t)(\fK(\II_n)) \subset \fK(\II_n)$
and one has $\alpha_n(t)^{-1} = \alpha_n(-t)$,
$t \in \RR$. 

We pick now any $C \in \fK(\II_n)$ and consider 
the familiar Dyson series  
\begin{align} \label{e4.2}
& \gamma_n(t)(C) \notag \\ 
& \ = C + 
\sum_{l=1}^\infty i^l \! \int_0^{t} \! ds_l \int_0^{s_l} \! ds_{l-1} \cdots
\int_0^{s_2} \! ds_1 [ \dots [C, \biV_n(s_1)] \dots , \biV_n(s_l) ]  \, .
\end{align}
Here we have introduced the short hand notation 
$\biV_n \doteq \sum_{j \neq k} V(\biQ_j - \biQ_k)$ and put \   
\mbox{$\biV_n(s) \doteq \alpha_n^{(0)}(s)(\biV_n)$}, $s \in \RR$.  
The integrals are defined in the strong operator topology on 
$\cH_n$ and the sum \eqref{e4.2} converges absolutely in norm, uniformly
on compact subsets of $t \in \RR$, because $\biV_n$ is a bounded operator. 
Adopting arguments from \cite{Bu3}, we obtain the following
result for the derivations appearing in this expansion.

\begin{lemma} \label{l4.1}
Let $s \mapsto \bdelta_n(s)$, $s \in \RR$, be the function 
with values in the derivations on $\cB(\cH_n)$, given by 
$$
\bdelta_n(s)(B) \doteq  i \, [B, \biV_n(s)] \, ,  \quad B \in \cB(\cH_n) \, .
$$ 
Its primitive $t \mapsto \bDelta_n(t) \doteq \int_0^t \! ds \, \bdelta_n(s)$ 
is pointwise defined in the strong operator topology.
It maps $\fK(\II_n) \subset \cB(\cH_n)$ into itself, 
$\bDelta_n(t)(\fK(\II_n)) \subset \fK(\II_n)$, and  
is point\-wise norm continuous and bounded by \
$\| \bDelta_n(t)(B) \|_n \leq 2 \, |t| \, \| \biV_n \|_n \, \| B \|_n$, 
$t \in \RR$.
\end{lemma}

\noindent \textbf{Remark:} The bounded operators 
$\int_0^t \! ds \biV_n(s)$ are not contained in 
$\fK(\II_n)$. It is essential for this   
result that the C*-algebra $\fK(\II_n)$ 
is not simple (\ie has ideals). Outer bounded derivations can 
and do exist in such cases. 

\begin{proof}
Since $\biV_n(s)$, $s \in \RR$, is a finite sum of
operators, it suffices to establish the statement for its summands
which are given by, $j, k \in \II_n$ and $j \neq k$, 
$$
V_{j k}(s) \doteq 
\alpha_n^{(0)}(s)(V(\biQ_j - \biQ_k))
= V(\biQ_j - \biQ_k + 2s(\biP_j - \biP_k)) \, , \quad 
s \in \RR \, .
$$
These operators are elements of the C*-algebra
$\fR_n(j \frown k)$, which is generated by all   
continuous functions, vanishing at infinity, of 
$\big( \bia \, (\biQ_j - \biQ_k) + \bib \, (\biP_j - \biP_k) \big)$ 
for $\bia, \bib \in \RR^s$. As a matter of fact,
this algebra coincides with 
the resolvent algebra generated by the resolvents of    
these linear combinations of canonically conjugate  
operators and it is faithfully
represented on $\cH_n$, cf.~\cite[Thm.~4.10]{BuGr2}. 
The algebra is stable under the action of 
the automorphisms $\alpha_n^{(0)}(u)$, $u \in \RR$, and it 
contains a compact   
ideal $\fC_n(j \frown k) \subset \fR_n(j \frown k)$ 
which is isomorphic to the algebra of compact 
operators on $L^2(\RR^s)$, cf.~\cite[Thm.~5.4]{BuGr2}. 

The first and vital step consists of the proof that the 
integrals $\int_0^t \! ds \, V_{j k}(s)$, defined in the strong operator 
topology, are elements of the compact ideal 
$\fC_n(j \frown k)$. A similar result was established in 
\cite[Lem.~2.1]{Bu3} for the case, where the present non-interacting 
Hamiltonian is replaced by the Hamiltonian of an isotropic 
harmonic oscillator. That  Hamiltonian also induces an 
automorphic action of time translations 
on the algebra $\fR_n(j \frown k)$. 
The argument given in \cite{Bu3} can be applied in the present case 
without major modifications and is put into  the appendix. 

Knowing that the integrals $\int_0^t \! ds \, V_{j k}(s)$, $t \in \RR$,
belong to the compact ideal \mbox{$\fC_n(j \frown k)$}, we 
can proceed and apply the same arguments as in \cite[Lem.~2.2]{Bu3}:  
we pick indices \mbox{$\II_m \subset \II_n$} 
and consider the corresponding algebra 
$\fC_n(\II_m) \subset \fK(\II_n)$. 
Given \mbox{$C \in \fC_n(\II_m)$}, there then appear four different 
possibilities for the action of the derivation $\Delta_{j k}(t)$ on $C$,
defined by the commutator with $\int_0^t \! ds \, V_{j k}(s)$. 

\noindent  \hspace*{3pt} (i) \ If both indices  $j,k \not\in  \II_m$, then 
 $V_{j k}(s)$ commutes with $C$ for any $s \in \RR$, and consequently 
$\Delta_{j k}(t)(C) = 0$ for $t \in \RR$. 

\noindent \hspace*{3pt} (ii) \  If both indices $j,k \in  \II_m$, then 
$\fR_n(j \frown k)$, hence $\fC_n(j \frown k)$, lies in the multiplier algebra 
of the algebra $\fC_n(\II_m)$.
Thus $\Delta_{j k}(t)(C) \in \fC_n(\II_m)$, $t \in \RR$. 

\noindent \hspace*{3pt} (iii) \ If 
$k \in  \II_m$, but  $j \not\in  \II_m$,
then 
$$
\Delta_{j k}(t)(C) \in [\fC_n(j \frown k), \, \fC_n(\II_m)] \, .
$$
As has been shown in \cite[Lem.~2.2]{Bu3}, these   
commutators are elements of the algebra $\fC_n(\II_{m+1})$, 
where $\II_{m+1} \doteq \II_m \overset{o}{\cup} j$ and 
the union symbol $ \overset{o}{\cup}$ indicates 
that the index $j$ is to be inserted at its proper place
within the ordered set $\II_{m+1}$. Thus
$\Delta_{j k}(t)(C) \in \fC_n(\II_{m+1})$, $t \in\RR$. 

\noindent \hspace*{3pt} (iv) \ Finally, if 
$j \in  \II_m$, but  $k \not\in  \II_m$,
then the same argument as in the preceding step shows that 
$\Delta_{j k}(t)(C) \in \fC_n(\II_{m+1})$, $t \in \RR$, where 
now $\II_{m+1} \doteq \II_m \overset{o}{\cup} k$
and $k$ has to be inserted at its proper place. 

Since the algebra $\fK(\II_n)$ is equal to the linear span
of the algebras $\fC_n(\II_m)$ for arbitrary ordered index 
sets  $\II_m \subset \II_n$ and $0 \leq m \leq n$, we
conclude that 
$\bDelta_n(t)(\fK(\II_n)) \subset \fK(\II_n)$,
$t \in \RR$. The remaining statements follow from the simple estimate
$\| (\bDelta_n(t_2) - \bDelta_n(t_1))(B) \|_n
\leq 2 \, \| \biV_n \|_n \, \| B \|_n \, | t_2 - t_1|$, $B \in \cB(\cH_n)$. 
This completes the proof of the lemma.
\qed \end{proof}

This lemma shows that the first nontrivial term in the 
expansion \eqref{e4.2} is contained in 
$\fK(\II_n)$. We will show by induction, similarly to the 
argument in \ \cite[Lem.~3.1]{Bu3}, 
that all other summands in this expansion are also contained in~$\fK(\II_n)$. 
There we rely on the following technical 
result.

\begin{lemma} \label{l4.2}
Let $s \mapsto D(s) \in \fK(\II_n)$ be norm 
continuous and let $s \mapsto \bdelta_n(s)$ be the derivations, 
defined in Lemma \ref{l4.1}.  Then 
$s \mapsto \bdelta_n(s)(D(s))$ is continuous in the strong
operator (s.o.) topology, $s \in \RR$. Its primitive, 
defined in the s.o.\ topology, has values in  $\fK(\II_n)$, 
$\int_0^t \! ds \,  \bdelta_n(s)(D(s)) \in \fK(\II_n)$,
$t \in \RR$. In fact, one has in the sense of norm convergence  
$$
\lim_{k \rightarrow \infty} \sum_{j = 1}^k
\big(\bDelta_n(jt/k) - \bDelta_n((j-1)t/k)\big) \, (D(jt/k))
= \int_0^t \! ds \,  \bdelta_n(s)(D(s))  \, .
$$ 
Here $t \mapsto \bDelta_n(t)$ is the primitive of 
$s \mapsto \bdelta_n(s)$, defined in
Lemma \ref{l4.1}, which maps $\fK(\II_n)$ 
into itself, $s \in \RR$. The 
function \ $t \mapsto \int_0^t \! ds \,  \bdelta_n(s)(D(s))$
is norm continuous and bounded with bound given by 
$\, 2 \, \| \biV_n \|_n \,  
\int_0^{|t|} \! ds \, \| D(s) \|_n$, \ $t \in \RR$.
\end{lemma}

\noindent \textbf{Remark:} The result holds also for maps  
$s \mapsto \bdelta_n(s)$ which are linear,
bounded, pointwise s.o. continuous, and have a
primitive $t \mapsto \bDelta_n(t)$ mapping $\fK_n$ into~itself.

\begin{proof}
Fixing $s_0 \in \RR$, we proceed to the splitting
$$
\bdelta_n(s)(D(s)) - \bdelta_n(s_0)(D(s_0)) =
(\bdelta_n(s) - \bdelta_n(s_0))(D(s_0)) +
\bdelta_n(s)(D(s) - D(s_0)) \, .
$$  
The first term on the right hand side vanishes in the limit
$s \rightarrow s_0$ in the s.o.\ topology since
$s \rightarrow \delta_n(s)$ is continuous in this 
topology, pointwise on~$\cB(\cH_n)$. The second term vanishes since 
$\bdelta_n(s)$ is uniformly bounded on bounded subsets of $\cB(\cH_n)$
and $D(s) \rightarrow D(s_0)$ in the norm topology.
Hence $\bdelta_n(s)(D(s)) \rightarrow \bdelta_n(s_0)(D(s_0))$
in the s.o.\ topology and the corresponding integrals 
$\int_0^t \! ds \, \bdelta_n(s)(D(s))$, $t \in \RR$, are well defined.
Let, without loss of generality, $t > 0$. By partitioning the
integration interval $0 \leq s \leq t$, we obtain the estimate
\begin{align} \label{e4.3} 
& \| \int_0^t \! ds \, \bdelta_n(s)(D(s)) - \sum_{j=1}^k
\int_{(j-1)t/k}^{jt/k} \! ds \, \bdelta_n(s)(D(jt/k)) \|_n  \nonumber \\
& = \| \sum_{j=1}^k 
\int_{(j-1)t/k}^{jt/k} \! ds \, \bdelta_n(s)(D(s) - D(jt/k)) \|_n \\
& \leq 2 \| \biV_n \| \, \sum_{j=1}^k  \int_{(j-1)t/k}^{jt/k} \! ds \,
\| D(s) - D(jt/k) \|_n \, . \nonumber
\end{align}
Because of the continuity properties of the function $s \mapsto D(s)$,
the upper bound tends to $0$ in the limit $k \rightarrow \infty$.
By assumption, $D(jk/k) \in \fK(\II_n)$ and, according to Lemma \ref{l4.1}, 
$\int_{(j-1)t/k}^{jt/k} \! ds \, \bdelta_n(s) = 
(\bDelta_n(jt/k) - \bDelta_n((j-1)t/k)$ maps  
$\fK(\II_n)$ into itself, $k  = 1, \dots , l$. Since 
$\fK(\II_n)$ is closed in the norm topology, this implies
$\int_0^t \! ds \, \bdelta_n(s)(D(s)) \in \fK(\II_n)$, $t \in \RR$. The
statement about the continuity 
and boundedness properties of the latter function follows  
easily from the fact that 
$\| \bdelta_n(s)(D(s)) \|_n \leq 2 \, \| \biV_n \|_n \, \| D(s) \|_n$, 
$s \in \RR$, completing the proof. \qed 
\end{proof}

We turn now to the analysis of 
the expansion \eqref{e4.2}. Given $l \in \NN$, the 
corresponding summand has the form,  
$C \in \fK(\II_n)$,  
$$
D_l(t)(C) \doteq \int_0^{t} \! ds_l \int_0^{s_l} \! ds_{l-1} \cdots
\int_0^{s_2} \! ds_1 \, \bdelta_n(s_l) \, \scirc \, \bdelta_n(s_{l-1})
\, \scirc \cdots \scirc \, \bdelta_n(s_1) \, (C) \, , 
$$
where the integrals are defined in the s.o.\ topology. 
For $l = 1$ it follows from Lemma \ref{l4.1} that 
$D_1(t)(C) \in \fK(\II_n)$, $t \in \RR$. 
Moreover, the function $t \mapsto D_1(t)(C)$ is norm
continuous and bounded by 
$\| D_1(t)(C) \|_n \leq 2 |t| \, \| \biV_n \|_n \| C \|_n$, $t \in \RR$. 
Given $l \in \NN$, \ the induction hypothesis 
consists of the assertion that $D_l(t)(C) \in \fK(\II_n)$, 
$t \mapsto D_l(t)(C)$ is norm continuous,  and 
$\| D_l(t)(C) \|_n \leq 2^l |t|^l/{\, l!}  \ \| \biV_n \|_n^l \, \| C \|_n$,
$t \in\RR$. For the induction step from $l$ 
to $l + 1$ we rely on Lemma \ref{l4.2}, noting 
that $D_{l + 1}(t)(C) = \int_0^t \! ds \, \bdelta_n(s)(D_l(s)(C))$,
where, by the preceding induction hypothesis, 
$s \mapsto D_l(s)(C) \in \fK(\II_n)$ 
is norm continuous and bounded.
According to Lemma \ref{l4.2} this implies that
$D_{l + 1}(t)(C) \in \fK(\II_n)$, 
$t \mapsto D_{l + 1}(t)(C)$ is norm continuous and 
\begin{align*}
\| D_{l+1}(t)(C) \|_n 
& \leq 2 \, \| \biV_n \| 
\int_0^{|t|} \! ds \, 2^l|s|^l / l! \  \| \biV_n \|_n^l \| C \|_n \\
& =   2^{l + 1}|t|^{l + 1} / (l+1)! \  \| \biV_n \|_n^{l + 1} \| C \|_n \, ,
\quad t \in \RR \, .
\end{align*}
This completes the induction. In view of the absolute 
convergence of the Dyson series \eqref{e4.2} in the norm topology, 
we have thus established the following fact.

\begin{lemma} \label{l4.3} 
Let $n \in \NN$. The Dyson maps $\gamma_n(t)$, $t \in \RR$,
defined in relation \eqref{e4.1}, map the subalgebra  
$\fK(\II_n) \subset \cB(\cH_n)$ onto itself, \ie they are 
automorphisms of this algebra. Moreover, the function 
$t \mapsto \gamma_n(t) \upharpoonright \fK(\II_n)$ is pointwise 
norm-continuous. 
\end{lemma}

It is now easy to show that the symmetric subalgebra 
$\fK_n \subset \fK(\II_n)$ of interest here 
is stable under the interacting dynamics. At this
point we make use of the fact that the Hamiltonians $H_n$, 
which were extended from equation \eqref{e2.1} to $\cH_n$,
commute with the unitaries $U_n(\pi)$,
$\pi \in \Sigma_n$. Hence the Dyson maps $\gamma_n(t)$ 
of $\fK(\II_n)$ commute with the operation of symmetrization,
$t \in \RR$. It therefore follows from the preceding lemma that 
$\gamma_n(t)(\fK_n) = \fK_n$, $t \in \RR$. 
Bearing in mind the properties of the non-interacting 
time evolution,  this implies 
$$
\alpha_n(t)(\fK_n) =
\alpha_n^{(0)}(t) \, \scirc \, \gamma_n(-t)(\fK_n) = \fK_n \, ,
\quad t \in \RR \, .
$$
As has been explained, 
the action of $t \mapsto \alpha_n^{(0)}(t)$ on $\fK_n$
is pointwise norm continuous. So we arrive at the following proposition. 

\begin{proposition} \label{p4.4}
Let $n \in \NN$. Given any dynamics, defined by a Hamiltonian $H_n$ 
of the form~\eqref{e2.1},
one has $\alpha_n(t)(\fK_n) = \fK_n$. Moreover, the 
function $t \mapsto \alpha_n(t)$ is pointwise norm
continuous on $\fK_n$, $t \in \RR$.
\end{proposition}

Having determined the properties of the action of the dynamics 
$\alpha_n(t)$, $t \in \RR$, on the restricted algebra  
$\overline{\bfA} \upharpoonright \cF_n = \fK_n$, $n \in \NN_0$,
we turn now to the full algebra. To this end we consider
the unitary operators $e^{itH}$, $t \in \RR$, on 
Fock space $\cF$, 
where $H$ is of the form given in equation~\eqref{e1.1}.
Their adjoint action on  $\cB(\cF)$ is denoted 
by $\balpha(t) \doteq \mbox{Ad} \, e^{itH}$, 
$t \in \RR$. Recalling that $H$ commutes 
with the particle number operator, we have 
$e^{itH} \upharpoonright \cF_n = e^{itH_n}$, where $H_n$ if given 
in equation~\eqref{e2.1}. Moreover, 
$\balpha(t) \upharpoonright \cB(\cF_n) = \alpha_n(t)$,
$t \in \RR$, where $\cB(\cF_n)$ is embedded into 
$\cB(\cF)$ by putting $\cB(\cF_n) \upharpoonright \cF_m = 0$
for $m \neq n$.

Given $A \in \overline{\bfA}$, it follows from Theorem \ref{t3.5}
that $\rho_n(A) = K_n \in \cK_n$.    
By Proposition \ref{p4.4} we have  
$\alpha_n(t)(K_n) \in \fK_n$ for $t \in \RR$ and any $n \in \NN_0$. 
Thus the sequence $\{\alpha_n(t)(K_n)  \}_{n \in \NN_0}$
defines for each $t \in \RR$ some element of the directed system 
$(\fK_n, \epsilon_n)_{n \in \NN_0}$. It is not clear from the 
outset that this sequence complies with the coherence condition, \ie defines
an element of the inverse limit $\bfK$. (Note that by modifying
arbitrarily the Hamiltonians $H_n$ for different $n \in \NN$, one  
still obtains sequences of operators in the directed system.)
In the subsequent lemma we will show, however, that the Hamiltonians defined 
in equation \eqref{e1.1} comply with the coherence conditon.
The desired result about the stability of the algebra~$\obfA$ 
under the corresponding dynamics then follows from Theorem~\ref{t3.5}. 

\begin{lemma} \label{l4.5}
Let $\balpha(t)$, $t \in \RR$, be the 
one-parameter group of automorphisms on $\cB(\cF)$, 
fixed by a Hamiltonian as in relation \eqref{e1.1}, and let 
$\alpha_n(t)$, $t \in \RR$, be its restriction 
to $\fK_n \subset \cB(\cF_n)$, $n \in \NN_0$. 
Then 
$$
\kappa_n \, \scirc \, \alpha_n(t) =  \alpha_{n-1}(t) \, \scirc \, \kappa_n 
\, , \quad n \in \NN \, .
$$
Thus $\balpha(t)$, $t \in \RR$, induces an action on the 
inverse system $(\fK_n, \kappa_n)_{n \in \NN}$ which preserves its  
inverse limit $\bfK$. 
\end{lemma}

\begin{proof}
Given $n \in \NN$, we consider first the restriction of
the non-interacting dynamics $\alpha^{(0)}_n(t)$, $t \in \RR$,
to the algebra $\fK_n$. We recall that $\fK_n = 
\sum_{m = 0}^n \fC_{m,n}$, where $\fC_{m,n}$ is the symmetrized
tensor product of the algebra of compact operators $\fC_m$ on 
$\fF_m$ and of $n-m$ unit operators. The algebra $\fC_m$ in turn
coincides with the symmetrized $m$-fold C*-tensor product of the 
algebra $\fC_1$ of compact operators on $\cF_1$.
Taking into account that the non-interacting dynamics does not mix
tensor factors, this gives 
$$
\alpha^{(0)}_n(t) \upharpoonright (\fC_m \otimes_s 
\underbrace{1 \otimes_s \cdots \otimes_s 1}_{n-m}) =
(\alpha^{(0)}_m(t) \upharpoonright \fC_m ) \otimes_s 
\underbrace{1 \otimes_s \cdots \otimes_s 1}_{n-m} \, ,
\quad 0 \leq m \leq n \, .
$$
Since the adjoint action of unitary operators maps compact operators
into compact operators, the statement then follows for the non-interacting 
dynamics from the definition of the inverse maps $\kappa_n$, 
$n \in \NN_0$. 

In view of this observation it suffices to establish the  
modified statement of the lemma, where 
the automorphisms $\alpha_n(t)$ 
are replaced by the Dyson maps $\gamma_n(t)$, $t \in \RR$, 
defined in relation \eqref{e4.1}.
As a matter of fact, it suffices to show that the 
derivations  $\bDelta_n(t) = \int_0^t \! ds \, \bdelta_n(s)$,
defined in Lemma \ref{l4.1} and underlying the 
Dyson series \eqref{e4.2} of $\gamma_n(t)$, satisfy
\begin{equation} \label{e4.4}
\kappa_n \, \scirc \, \bDelta_n(t) 
= \bDelta_{n-1}(t) \, \scirc \, \kappa_n  \, , \quad n \in \NN \, ,
\end{equation}
where we put $\bDelta_0(t) \doteq 0$, $t \in \RR$. 

Before going into the proof of this relation, let us 
explain why it implies the result. 
As shown in the proof of Lemma \ref{l4.3},
the $l$th term in the Dyson series, obtained by 
iteration, has the form
$D_l(t) = \int_0^t \! ds \, \bdelta_n(s) (D_{l-1}(s))$,
where $s \mapsto D_{l-1}(s) \in \fK_n$ is norm continuous, $l \in \NN$. 
Moreover, it was shown in Lemma~\ref{l4.2} that these
integrals, defined in the strong operator topology, 
can be approximated in norm by the sums 
$$
\sum_{j = 1}^k
\big(\bDelta_n(jt/k) - \bDelta_n((j-1)t/k)\big) \, (D_l(jt/k)) \, .
$$
Anticipating the asserted relation, this implies 
\begin{align*}
& \kappa_n \big( \sum_{j = 1}^k
\big(\bDelta_n(jt/k) - \bDelta_n((j-1)t/k)\big) \, (D_l(jt/k)) \big) \\
& = \big( \sum_{j = 1}^k
\big(\bDelta_{n-1}(jt/k) - \bDelta_{n-1}((j-1)t/k)\big) \, \kappa_n(D_l(jt/k)) ) 
\, .
\end{align*}
Bearing in mind the norm continuity of the maps 
$\kappa_n : \fK_n \rightarrow \fK_{n-1}$,
one therefore arrives 
in the limit $k \rightarrow \infty$  with the help of Lemma \ref{l4.2}  
at the equality 
$$
\kappa_n(D_l(t)) = \kappa_n\big(\int_0^t \! ds \, \bdelta_n(s) (D_{l-1}(s) \big)
= \int_0^t \! ds \, \bdelta_{n-1}(s) (\kappa_n(D_{l-1}(s))) \, .
$$ 
One can apply now  the same argument to the function 
$s \mapsto D_{l-1}(s)$ on the right hand side
of this equality. Upon iteration, one sees that the action of 
$\kappa_n$ maps the multiple integrals, underlying the expansion 
of $\gamma_n(t)$ (acting on $\fK_n$)  
into the multiple integrals underlying
the expansion of $\gamma_{n-1}(t)$ (acting on $\fK_{n-1}$). 
The statement then follows from the convergence of the
Dyson expansions. 

Let us turn now to the proof of relation \eqref{e4.4}.
There we apply again the arguments used in the proof of Lemma \ref{l3.4}
on which the definition of the maps $\kappa_n$ was based.
To this end we need to reexpress the
derivations in relation \eqref{e4.4} by the underlying   
global operator, involving creation and annihilation 
operators. Thus we go back from the potentials
$\biV_n$, $n \in \NN$, to 
$$
\biV \doteq \int \! d\bix \! \! \int \! d\biy \  
 a^*(\bix) a^*(\biy) V(\bix - \biy) a(\bix) a(\biy) \, . 
$$
This will enable us to make use of locality properties of the operators. 
In a first step, we consider two-body potentials $V$ which are 
given by functions
with compact support. Then there exist families of operators 
$A \in \bfA$ such that their restrictions to $\cF_n$ are dense in $\fK_n$ and 
the commutator $[\biV, \alpha^{(0)}(s)(A)]$ can be approximated in 
norm by localized operators. Using the fact that 
$$
[\biV(s), A] = 
\balpha^{(0)}(s)( \, [\biV, \alpha^{(0)}(-s)(A)] \, )
$$ 
and applying the adjoint
of the leftmost automorphim $\balpha^{(0)}(s)$ to the states, 
the arguments given in Lemma \ref{l3.4} imply 
\begin{equation} \label{e4.5}
\lim_{\sbix \rightarrow \infty} \langle \bPsi^n(\bix), \, [\biV(s) , A] \, 
\bPhi^n(\bix) \rangle
= 1/n \ \langle \bPsi^{n-1}, \, [\biV(s) , A] \, \bPhi^{n-1} \rangle \, .  
\end{equation}
Here the wave functions $f_n, g_n$ in the lemma are chosen to 
coincide and to be normalized. By the dominated convergence
theorem, the same result is obtained for the integrated operators
$\int_0^t \! ds \, \biV(s)$, $t \in \RR$. Reexpressing the resulting
equality by the initial quantities, we arrive at
$$
\lim_{\sbix \rightarrow \infty} \langle \bPsi^n(\bix), \, 
\bDelta_n(t)(\rho_n(A)) \, \bPhi^n(\bix) \rangle
= 1/n \, \langle \bPsi^{n-1}, \, \bDelta_{n-1}(t)(\rho_{n-1}(A)) \, 
\bPhi^{n-1} \rangle \, . 
$$
Since $\kappa_n \, \scirc \rho_n = \rho_{n-1}$, this 
implies $\kappa_n \scirc \, \bDelta_n(t) = \bDelta_{n-1}(t) \scirc \, \kappa_n$
on a dense set of operators 
in $\fK_n$, whence, by the continuity propertis
of the maps, on all of $\fK_n$. Finally, the restrictions on 
the two-body potential $V$ 
can be removed due to the norm continuity of the derivations 
$\bDelta_n(t)$ with regard to
changes of the potential. So the statement follows.

It remains to exhibit the required operators $A \in \bfA$ and to 
control the localization properties and the 
norm of the commutators 
$[\biV, \balpha^{(0)}(s)(A)]$ on $\cF_n$, $n \in \NN$.
To this end we choose operators $X_n(f)$, $f \in \cD(\RR^s)$,
of the type used in the proof of Lemma \ref{l3.3}. Given 
any non-zero function $f$, they have the form 
$X_n(f) \doteq E_{[0,n]}(f) \, a(f)$, where $E_{[0,n]}(f)$
is the spectral projection of the number operator 
$N(f) \doteq \| f \|_2^{-2} \, a^*(f) a(f)$ corresponding to the  
spectrum in $[0,n]$. It was shown in Lemma \ref{l3.3} that the 
sums of products of the gauge invariant 
combinations $X_n(g)^* X_n(f) \upharpoonright \cF_n$ with  
$f,g \in \cD(\RR^s)$ are dense in~$\fK_n$.  

What matters here is the fact that the maps 
$f \mapsto  X_n(f) \upharpoonright \cF_n$
are continuous with regard to the norm on $L^2(\RR^s)$,  provided 
the functions $f$ are staying away from $0$. This is obvious
for the function $f \mapsto a(f) \upharpoonright \cF_n$,
where one has 
$$
\|(a(f) - a(g))\|_n \leq n^{1/2} \, \| f - g \|_2 \, . 
$$
In order to give an estimate for  
$\| E_{[0,n]}(f) - E_{[0,n]}(g) \|_n$, it is convenient 
to proceed to the resolvents of the underlying 
number operators. There 
one obtains for complex $z \not\in \NN_0$
\begin{align*}
\|(z - & N(f))^{-1} - (z - N(g))^{-1} \|_n \\
& \leq 1/d_\NN(z)^{2} \ \| N(f) - N(g) \|_n  
\leq 2 n/d_\NN(z)^{2} \ \big\| f/\|f\|_2 - g/\|g\|_2 
\big\|_2 \, 
\end{align*}
where $d_\NN(z)$ is the distance between $z$ and $\NN_0$. 
A similar estimate results for the difference of the projections by
performing a contour integration of the resolvents. So, summing
up, we arrive at 
$$
\|X_n(f) - X_n(g)\|_n \leq c_n \, 
\big(\|f\|_2 / \| g \|_2 + \| g \|_2 / \| f \|_2 \big) \ \|f - g \|_2 \, ,
$$
where $c_n$ depends only on $n$. This estimate shows that one can 
approximate any operator $X_n(f)$, $f \in L^2(\RR^s)$, in the 
given norm by operators $X_n(f_r)$, where $f_r$ are test functions
having support in a ball $\BB_r \subset \RR^s$
about $0$ with sufficiently large radius $r$.
We shall say that $X_n(f_r)$ is localized in $\BB_r$.  

This approximation procedure can also be applied to
arbitrary sums and products of the gauge invariant 
combinations $X_n(f)^*X_n(g)$, $f,g \in L^2(\RR^s)$.
It shows that for any operator $A$ in this family there exist 
members $A_r$ which are localized in~$\BB_r$ 
and $\| A - A_r \|_n \rightarrow 0$ for $r$ tending to infinity.
The action of the time translations $\balpha^{(0)}(s)$
on the operators 
does not affect this feature   
since it only induces isometries $f \mapsto f(s)$ of the
functions underlying the construction of $A$. Therefore, 
there are still operators~$A_r$, depending also on $s$, 
which are localized in $\BB_r$, and 
$\| \balpha^{(0)}(s)(A) - A_r \|_n \rightarrow 0 $
in the limit of large $r$. So the property of 
having localized approximants is stable under the
non-interacting time evolution. Adopting 
the terminology used in relativistic quantum field theory, one
may say that the operators in this family are quasi-local.  

Let us consider now the commutators
$[\biV, A]$, where $A$ belongs to the family of operators 
just constructed. One has 
$\|  [\biV, A] \|_n \leq 
2 n(n-1) \, \| V \| \, \| A \|_n$.  Moreover, for given $A$ and $s \in \RR$,  
there exist operators $A_r$ which are localized 
in $\BB_r$ and 
$\|  [\biV, (\balpha(-s)(A) - A_r \big)] \|_n \rightarrow 0$
for $r \rightarrow \infty$. 
Since $V$ has compact support, the operator
$[\biV, A_r]$ commutes with all operators $a^*(h)$,
where $h$ has support in the complement of a Ball $\BB_R$,
containing the region $\BB_r + \text{supp} \, V$. 

In the computation of the matrix elements of these operators, 
there appear now the time translated 
wave functions $h(-s)$ of the given vectors.
Again, there exist 
test functions $h_r$, having support in $\BB_r$, 
and $\| h(-s) - h_r \|_2 \rightarrow 0$ for $r \rightarrow \infty$. 
Hence 
$\| a^*(h(-s)) - a^*(h_r) \|_{n-1} = \| a(h(-s)) - a(h_r) \|_n
\rightarrow 0 $. Since the spatial translations commute with the
free time evolution and do not affect the norm, 
the same relation holds for the 
translated wave functions $h(-s,\bix)$, $h_r(\bix)$, uniformly
for $\bix \in \RR^s$. So one can replace the operators 
$a^*(f_n(-s,\bix))$, creating the particle component 
of the vectors which is shifted
away, by the operators $a^*(f_{n,r}(\bix))$. The latter 
operator is localized in $\BB_r + \bix$ and therefore commutes 
with $[\biV, A_r]$ for large translations $\bix$. 

One can now apply the arguments in Lemma \ref{l3.4} and thereby 
arrives at relation~\eqref{e4.5} for the approximating operators
$A_r$. Since the matrix elements are uniformly bounded 
in $s$ one also obtains the corresponding relation for its integrated version.
Recalling that the family of operators $\rho_n(A_r)$ is dense
in $\fK_n$, this proves relation~\eqref{e4.4} for 
two-body potentials $V$ with compact support. Since 
$\bDelta_n(t)$ is linear in~$V$ and
$\| \bDelta_n(t)(K_n) \|_n \leq 2 n (n-1) \| V \| \| K_n \|_n$,
$K_n \in \fK_n$, this result extends to arbitrary potentials
$V \in C_0(\RR^s)$ and $n \in \NN$, completing the proof.  
\qed \end{proof}

The preceding results clarify the physical meaning of the inverse maps
$\kappa_n$. They describe the operation of removing from $\rho_n(\bfA)$
all observables which are sensitive to a particle which is located very far
away and therefore does not interact anymore with the remaining particles. 
One thereby arrives at the observables in the representation $\rho_{n-1}$
and the corresponding dynamics, $n \in \NN$. On the mathematical 
side, this implies that the dynamics $\balpha$ is 
compatible with the structure of the inverse limit $\bfK$. This leads us 
to the main result of this section.

\begin{theorem} \label{t4.6} 
Let $\balpha \doteq \{\balpha(t)\}_{t \in \RR}$ 
be the group of automorphisms of $\cB(\cF)$,
defined by a Hamiltonian of the form given in \eqref{e1.1}. 
These automorphisms map the subalgebra 
$\overline{\bfA} \subset \cB(\cF)$ onto itself, \ie they are 
automorphisms of $\overline{\bfA}$. Moreover, they 
act pointwise continuously on it with regard to the 
seminorms $\| \, \cdot \, \|_n$, $n \in \NN_0$. 
For any $\balpha$ there is a (in the latter topology) dense 
subalgebra $\overline{\bfA}_{\, \alpha} \subset \overline{\bfA}$
on which~$\balpha$ acts pointwise norm continuously
with regard to the C*-norm $\| \, \cdot \, \|_\infty$. Thus
$(\overline{\bfA}_{\, \alpha}, \balpha)$ forms a C*-dynamical system
for the given dynamics. 
\end{theorem}
\begin{proof}
It was shown in Lemma \ref{l4.5} that $\balpha(t)$, $t \in \RR$, defines
a group of automorphisms of the inverse limit $\bfK$ of 
$(\fK_n, \kappa_n)_{n \in \NN}$. The algebra $\bfK$ in turn is isomorphic to
$\overline{\bfA}$ according to Theorem \ref{t3.5}, proving the
first part of the statement. The stated continuity 
properties of theses automorphisms with regard to the 
topology induced by the seminorms then follow from
Proposition \ref{p4.4}. These continuity properties imply in 
particular that the mollified operators 
$\int \! ds \, f(s) \balpha(s)(A)$, defined in the strong operator
topology for $f \in L^1(\RR)$, $A \in \overline{\bfA}$, generate
a unital subalgebra $\overline{\bfA}_{\, \alpha} \subset \overline{\bfA}$
which is dense in $\overline{\bfA}$ with regard to the seminorms. 
For, the restricted integrals 
$\int \! ds \, f(s) \balpha(s)(A) \upharpoonright \cF_n$ are defined
in the norm topology and form a coherent sequence, $n \in \NN_0$.  
The stronger
continuity properties of the elements of $\overline{\bfA}_{\, \alpha}$
with regard to the action of the automorphisms then follow from the
estimate \ $\| \int \! ds \, f(s) \balpha(s)(A) \|_\infty
\leq \int \! ds \, |f(s)| \, \| A \|_\infty $ and the continuity
properties of the elements of $L^1(\RR)$ with regard to translations.
\qed \end{proof}

We conclude this section by discussing locality properties of 
the observables which entered already in preceding lemmas. 
They are a distinctive feature of  
field theory, having no counter part in the particle picture. 
The resolvent algebra $\bfR$ is, by construction, the 
C*-inductive limit of the net of its subalgebras $\bfR(\biO)$ 
based on the open, bounded  regions $\biO \subset \RR^s$. 
The algebras $\bfR(\biO)$ are generated by the 
subsets of resolvents 
$R(\lambda, f)$, where $\lambda \in \RR \backslash \{ 0 \}$
and $f \in \cD(\biO)$, the space of test functions having support in $\biO$. 
Since the algebras $\bfR(\biO)$ are stable under gauge 
transformations, we can proceed to their gauge invariant
subalgebras \mbox{$\bfA(\biO) \subset \bfR(\biO)$}, \ $\biO \subset \RR^s$,
and the algebra of all gauge invariant observables $\bfA$
is the C*-inductive limit of these subalgebras. It is an
immediate consequence of the canonical commutation relations
that elements of algebras assigned to disjoint regions in 
$\RR^s$ commute with each other, thereby implementing
the principle of locality (statistical independence of
spatially separated observables) at the kinematical level.

These locality properties carry over to the algebras 
$\overline{\bfA}(\biO)$, $\biO \subset \RR^s$, obtained by 
extending $\bfA(\biO)$.
Its elements consist of bounded operators $A$ on $\cF$
whose restrictions to all subspaces with limited 
particle number coincide with elements of $\bfA(\biO)$.
Again, the algebra $\overline{\bfA}$ 
is an inductive limit of its local subalgebras, but now with regard to 
the topology given by the seminorms \mbox{$\| \, \cdot \, \|_n$}, 
$n \in \NN_0$. In fact, given $A \in \overline{\bfA}$, 
there exists a sequence of operators $\{ A_n \in \bfA(\biO_n) \}_{n \in \NN}$, 
localized in bounded regions  $\biO_n$, 
such that \ $\lim_{n \rightarrow \infty} \|A_n - A \, \|_m = 0$
for any given $m \in \NN_0$. This sequence is obtained by 
approximating $A \upharpoonright \cF_n$ in norm 
by local operators $A_n \upharpoonright \cF_n$, 
which is possible by the definition of $\overline{\bfA}$
and the inductive structure of $\bfA$. 
These approximations have to be progressively improved 
for increasing $n$. Taking into account that the 
semi-norms $\| \, \cdot \, \|_m$ are monotonically increasing
with $m \in \NN_0$, one thereby obtains the desired 
approximating sequence $A_n$, satisfying 
$\| A - A_n \|_m \rightarrow 0$ in the limit of large $n$
for any $m \in \NN_0$.

Whereas the non-relativistic dynamics does not preserve the 
kinematical local structure of the observables, our results 
imply that there remain some quasi-local properties of the 
time translated observables in the sense that these observables
still  
commute at asymptotic spatial distances. We already made use of this fact
in case of the non-interacting dynamics in the proof of Lemma \ref{l4.5}.
But it applies to the interacting case, as well. Let $\balpha(\bix)$,
$\bix \in \RR^s$, be the automorphism 
group of spatial translations on $\bfA$ which acts on the 
generating resolvents according to 
$\balpha(\sbix) (R(\lambda, f)) \doteq R(\lambda, f(\bix))$,
where $f(\bix)$ denotes the test function $f$, translated by $\bix$. 
These automorphisms commute with the time translations 
$\balpha(t)$, $t \in \RR$, and can be extended to the algebra
$\overline{\bfA}$. Hence, 
putting $\balpha(t,\bix) \doteq \balpha(t) \, \scirc \, \balpha(\bix)$
and taking into account the preceding remarks about the inductive
structure of $\obfA$, we obtain  
the following corollary of Theorem \ref{t4.6}.

\begin{corollary} \label{c4.7}
Let $A, B \in \overline{\bfA}$ and let $t \in \RR$. Then 
$$
\lim_{\sbix \rightarrow \infty} \, 
\| \, [ \balpha(t, \bix)(A), B ] \, \|_n
= 0 \, , \quad n \in \NN_0 \, .
$$
\end{corollary}

So the time translated observables are still quasi-local in this sense. 
It is an interesting question whether for given 
interaction potential one can establish 
more specific bounds for these commutators, \eg  
of Lieb-Robinson type \cite{LiRo}, indicating   
limitations on the speed of propagation of causal effects.

\section{States}
\setcounter{equation}{0}

Having settled the framework, we turn now to its applications 
for the construction and interpretation of states in many body theory. 
It is the primary purpose of this part of our article 
to indicate how the algebraic approach sheds new light on some
topics of physical interest; functional analytic details are largely omitted. 

\vspace*{-2mm}
\subsection{Ground states and infra vacua} \label{s5.1} 

The pair potentials $V$ considered here may lead to 
bound states whose  
energy lies below the energy of the vacuum state $\Omega$, which is
put equal to $0$ by the definition of the Hamiltonians $H$ in 
equation~\eqref{e1.1}. This defect can be resolved by renormalizing the
Hamiltonians.
Let $H$ be given and let $-E(n)$ be the infimum of the spectrum of 
$H \upharpoonright \cF_n$, which exists for the class
of potentials considered here, $n \in \NN_0$. With the help of 
the particle number operator $N$, one then defines 
the renormalized Hamiltonian $H_r \doteq H + E(N)$. It is non-negative,
the vacuum $\Omega$ is its ground state, and it 
induces the same time evolution on 
$\overline{\bfA}$ as the given Hamiltonian~$H$
since $N$ commutes with the elements of this algebra. 
A particularly nice class of potentials $V$ are those of positive type 
(having non-negative Fourier transforms). They can have bound 
states as well.
One easily verifies for these potentials that the corresponding 
renormalized Hamiltonians $H_r \doteq H + V(0) N$ are non-negative,
where the value $V(0) > 0$ of the potential at the
origin resembles a chemical potential. We will return to these 
potentials in our discussion of equilibrium states.

A more interesting class of states are approximate ground states, 
formed by an infinity of low energy Bosons, such as 
infra-vacua or condensates. These states cannot be described by
vectors in Fock space, they lead to inequivalent representations
of the algebra $\overline{\bfA}$. Such states are obtained in
the present setting by forming suitable sequences of vectors
in Fock space which determine sequences of expectation functionals
on $\overline{\bfA}$. The limits of these functionals yield 
via the GNS-construction the desired representations. 
In more detail, one proceeds from sequences of states
with increasing particle number, retaining control of 
their localization properties and their total energy. 
For simplicity we consider two-body potentials $V$ which are 
non-negative and of short range. 
Given $n \in \NN$, one defines vector states in 
$\cF_n \subset \cF$ of the form
$$
\Psi_{L,n} \doteq (n!)^{-1/2} \, 
f_L \otimes_s \cdots \otimes_s f_L \, .
$$
Here $\bix \mapsto f_{L}(\bix) \doteq L^{-s/2} f(\bix/L)$
are normalized elements of the single particle space $L^2(\RR^s)$.
In order to gain control on the total energy, one
interprets the vectors $\Psi_{L,n}$ as asymptotic configurations of  
scattering states of particles. To this end one 
considers on $\cF_n$ the (free channel) M{\o}ller wave operators 
$\widehat{\Omega}_n$, which for the potentials considered here 
are limits of the operators $e^{itH_n} e^{-itH_{0 n}}$ at asymptotic times
$t \rightarrow \infty$, cf.~\cite{DeGe}. With their help one
then proceeds to the outgoing scattering states
$\widehat{\Psi}_{L,n} \doteq \widehat{\Omega}_n \, \Psi_{L,n}$. 
What matters here is the fact that the isometric M{\o}ller operators 
intertwine the interacting and non-interacting dynamics, 
$H_n \, \widehat{\Omega}_n = 
\widehat{\Omega}_n \, H_{0 n}$. 
Hence one obtains for the expectation value of the full energy operator 
the bounds 
$$
0 \leq \langle \widehat{\Psi}_{L,n}, \, H_n \, \widehat{\Psi}_{L,n} \rangle
= \langle \Psi_{L,n}, \, H_{0,n} \, \Psi_{L,n} \rangle
= n L^{-2} \int \! d \bix \, |\bpartial f(\bix) |^2 \, .
$$
The functionals \ 
$\widehat{\omega}_{L,n}(\, \cdot \,) \doteq 
\langle \widehat{\Psi}_{L,n}, \, \cdot \ \widehat{\Psi}_{L,n} \rangle$
on the C*-dynamical system $(\overline{\bfA}_{\, \alpha}, \balpha)$
retain finite total energy in the limit $n \rightarrow \infty$, 
$n L^{-2} = \mbox{const}$. Note that 
the integral can be made arbitrarily small
for suitable choices of $f$, so in a sense the
states may be viewed as approximate ground states. 
Moreover, choosing functions whose Fourier transforms 
have compact support, one has control on the
properties of the weak-*-limit points of these functionals;
such limit points always exist by compactness arguments 
(Banach-Alaoglu theorem). Making use of the fact that
one is dealing with C*-dynamical
systems, one can show that for any such function $f$ all limit 
functionals lead
via the GNS-construction to positive energy representations. 
In other words, the automorphisms inducing the 
time translations are unitarily implemented in these 
non-Fock representations  
and have positive selfadjoint generators, cf. \cite[Sec. II.5]{Bo}.

We conclude this section by noting that the 
states \ $\widehat{\omega}_{L,n}(\, \cdot \,)$  
describe at asymptotic times Bose-Einstein
condensates with particle density $\bix \mapsto n L^{-s} |f(\bix/L)|^2$. 
It is, however, not clear whether they have an interpretation
as condensates also at finite times. In view of the interest in this 
phenomenon, cf.\ \cite{CoDeZi,LiSeSoYn,Ve} and references 
quoted there, it seems worth while to explore this question
in more detail also from the present algebraic point of view.

\subsection{Particle  observables}
\setcounter{equation}{0}

As we have seen, the 
algebra of field theoretic observables $\overline{\bfA}$ does
not contain operators having continuous spectrum. In  
particular, it does not contain particle observables such as 
the momentum operators. In order to uncover from the algebra 
such observables one has to proceed to 
observations at asymptotic times. For
then the interaction between the fields fades away and the   
sub-leading particle aspects can surface. This fact was 
established first by Araki and Haag in relativistic quantum field 
theory \cite{ArHa, Ar}. Their reasoning can be carried over to the
present non-relativistic setting and we briefly outline their 
construction. In order to simplify the 
discussion, we restrict our attention again to potentials $V$
which are non-negative and of short range. Then there are 
no bound states and since the theory is also asymptotically 
complete, cf.~\cite{DeGe} and references quoted 
there, the isometric M{\o}ller operators are invertible.

Let 
$\widehat{\Omega}_n$ be the unitary M{\o}ller operator, mapping the states
$\Psi_n \in \cF_n$ onto outgoing scattering states,
denoted by $\widehat{\Psi}_n \doteq \widehat{\Omega}_n \Psi_n$, $n \in \NN_0$. 
One has by the very definition of the M{\o}ller operators the 
asymptotic equality
$$
\lim_{t \rightarrow \infty} \, 
\big(\langle \widehat{\Psi}_n, \, 
\balpha(t)(A) \, \widehat{\Psi}_n \rangle -
\langle \Psi_n, \, \balpha^{(0)}(t)(A) \, \Psi_n \rangle \big) = 0 \, ,
\quad A \in \overline{\bfA} \, .
$$
Bearing in mind the structure of  
$\overline{\bfA} \upharpoonright \cF_n = \fK_n$ and the fact
that the compact tensor factors appearing 
in these operators  are mapped
onto themselves by the non-interacting dynamics, one can show,   
similarly to the proof of Lemma \ref{l3.4}, that 
$\lim_{\, t \rightarrow \infty} 
\langle \Psi_n, \, \balpha^{(0)}(t)(A) \, \Psi_n \rangle  =
\langle \Omega, \, A \, \Omega \rangle $. 
Since $\widehat{\Omega}_n$ maps each $\cF_n$ onto itself, $n \in \NN_0$,
this implies that for all normalized vectors $\Phi \in \cF$ one has 
$$
\lim_{t \rightarrow \infty} \, 
\langle \Phi, \, \balpha(t)(A) \, \Phi \rangle = 
\langle \Omega, \, A \, \Omega \rangle \, ,
\quad A \in \overline{\bfA} \, .
$$
So the asymptotically dominant contributions 
to the expectation values of the observables arise from the Fock vacuum. 
 
For the next to leading order one subtracts from the 
observables the asymptotically dominant contributions and proceeds to 
$A_0 \doteq A - \langle \Omega, \, A \, \Omega \rangle \, \one$,
where one restricts attention to operators
$A \in \overline{\bfA}(\biO)$ which are localized in bounded regions
$\biO \subset \RR^s$. The operators $\balpha(t,\bix)(A_0)$ 
register in expectation values in the region \ $\biO + \bix$ at time $t$
deviations from the Fock vacuum.
Because of the spreading of wave packets, these deviations 
tend to $0$ at asymptotic times. In order to compensate this effect
one integrates the operators \ $\balpha(t,\bix)(A_0)$ over regions
in space whose diameter increases with time. In the proof that 
this is meaningful one makes 
use of the specific structure of the operators in $\overline{\bfA}$ and
applies arguments given in \cite[Lem.~2.2]{Bu3}, where the role 
of the energy operator has   
to be replaced by the particle number operator. One thereby 
finds that for any test function $h$ and any compact
region $\biO \subset \RR^s$ the spatial averages of the operators 
$A \in \overline{\bfA}(\biO)$ satisfy the norm estimates    
$$
\| \int \! d\bix \, h(\bix) \, \balpha(\bix) (A_0) 
\upharpoonright \cF_n \|_n \leq n \, c_{A} \, \| h \|_\infty \, ,
\quad n \in \NN_0 \, ,
$$
where $\| h \|_\infty$ denotes the supremum norm of $h$
and the constant $c_A$ depends only on the chosen operator $A$. 
Similarly to the leading order, this bound and the definition 
of the M{\o}ller operators imply  
that one can proceed again to an asymptotic equality 
\begin{align*}
& \lim_{t \rightarrow \infty} \, 
\Big( \langle \widehat{\Psi}_n, \, \balpha(t)(\! \int 
\! d\bix \, h(\bix/t) \, \balpha(\bix) (A_0)) \, \widehat{\Psi}_n \rangle \\
& - \langle \Psi_n, \, \balpha^{(0)}(t)(\! \int 
\! d\bix \, h(\bix/t) \, \balpha(\bix) (A_0)) \, \Psi_n \rangle \Big) 
= 0 \, .
\end{align*}
By a routine computation of the second (non-interacting) term one gets  
\begin{align*}
& \lim_{t \rightarrow \infty} \, \langle \Psi_n, \, \balpha^{(0)}(t)(\! \int 
\! d\bix \, h(\bix/t) \, \balpha(\bix) (A_0)) \, \Psi_n \rangle \\
& = c_s \int 
\! d\bip \, h(2 \bip) \langle \bip | A_0 |  \bip \rangle \, 
\langle \Psi_n, \widetilde{a}^*(\bip) 
\widetilde{a}(\bip) \Psi_n \rangle \, .
\end{align*}
Here $c_s$ is some dimension dependent constant, 
$\widetilde{a}^*, \widetilde{a}$ are the 
Fourier transforms of the creation and annihilation operators, and 
$\bip, \biq \mapsto  \langle \bip | A_0 | \biq \rangle$ denotes 
the integral kernel of $A_0 \upharpoonright \cF_1$ in momentum 
space. 
This kernel is a continuous function for all compactly  
localized operators; its restriction to the 
diagonal $\bip = \biq$ is called sensitivity function 
and encodes the response of the observable $A_0$ to single particle
excitations, cf.~\cite{ArHa, Ar}. 
These findings can be combined into a single asymptotic formula, 
\begin{equation} \label{e5.1} 
\lim_{t \rightarrow \infty}  \int 
\! d\bix \, h(\bix/t) \, \balpha(t, \bix) (A_0) 
= c_s \int 
\! d\bip \, h(2 \bip) \langle \bip | A_0 | \bip \rangle \,
\widehat{a}^*(\bip) \widehat{a}(\bip) \, ,
\end{equation}
where the limit exists on the domain of the number 
operator $N$ and $\widehat{a}^*$, $\widehat{a}$ denote the outgoing
creation and annihilation operators in momentum space.
An analogous formula holds at negative asymptotic times.
This formula replaces the familiar asymptotic
condition for field operators in case of the observables, which 
do not change the particle number. 

The above formula shows in particular that the particle 
momenta, desribed by the operators 
$\widehat{M} = 
\int \! d\bip \, m(\bip) \, \widehat{a}^*(\bip) \widehat{a}(\bip)$, 
where the functions $m$ result from the sensitivity and averaging 
functions, become meaningful observables at asymptotic times.
They can be determined from the underlying algebra in a 
universal manner which does not depend on the dynamics. 
Other important particle properties, such as the  
collision cross sections, can likewise be determined  
along these lines by considering products of the above
operators, cf.\ the discussion in~\cite{ArHa, Ar, BuPoSt}. 

The preceding
results hold also for potentials admitting bound states. There the 
Fock space $\cF$ splits at asymptotic times into the 
tensor product of Fock spaces corresponding to the elementary
particle and to its stable bound states. The bound states 
are described in the same manner as the underlying particles   
in spite of their possibly complex internal structure.
This internal structure is encoded in their  
sensitivity functions $\bip \mapsto \langle \bip, b | 
A_0 | \bip, b \rangle$, where $b$ labels the bound
states. Also in those cases, the leading asymptotic contributions 
to the expectation values of observables are described by the Fock vacuum. 
But in next to leading order one has to replace in formula 
\eqref{e5.1} the integral on the right hand side by sums of similar terms 
containing the sensitivity functions
and the creation and annihilation operators
of the underlying particle and of all bound states in the theory.

The present results rely to some extent on well known 
facts about asymptotic completeness in quantum mechanics, cf.\  
\cite{DeGe} and references quoted there. Within the present
algebraic setting one may, possibly, give more direct
proofs of relation \eqref{e5.1}. For, one has 
quite specific information about the properties of the observables,
in particular about their close relation to compact operators.
It may therefore be possible to establish in the present concrete 
setting the existence of the limit in relation \eqref{e5.1} by a 
refinement of the Arveson spectral theory of spacetime 
automorphism groups~\cite{Arv}. What is needed 
is a finer characterization of the spectra of the automorphisms 
in analogy to the measure classes for Hilbert space operators.
In order to arrive at the desired convergence one would have to 
exclude in the present context the analogue of 
``singular continuous spectrum''. 
For some first steps into this direction,  cf.\ \cite{Dy}.

\subsection{Equilibrium states} \label{s5.3}

The conventional method of constructing equilibrium states in quantum field 
theory is based on the consideration of Gibbs-von Neumann ensembles in 
bounded regions (boxes). In the present algebraic framework one deals from
the outset with observables in infinite space, so one has to proceed 
differently. One replaces the sharp boundaries of a box by 
external confining forces which are conveniently described by a
harmonic oscillator potential. Thus one considers  
Hamiltonians of the form, $L > 0$, 
\begin{align*}
H_L  &  \doteq  
\int \! d\bix \big( \bpartial a^*(\bix) \, \bpartial a(\bix) 
+ (\bix^2/L^4) \, a^*(\bix) a(\bix) \big) \\
& + \int \! d\bix \! \! \int \! d\biy \ a^*(\bix) a^*(\biy)
V(\bix - \biy) a(\bix) a(\biy) \, ,
\end{align*}
where we assume here that the interaction potential $V$ is of positive type. 
Proceeding to the renormalized Hamiltonians  
$H_{L \, r} \doteq H_L + V(0) N $, the corresponding 
canonical ensembles can then be described by 
density matrices on Fock space $\cF$.
In fact, one has $H_{L \, r} \geq H_{0 L}$, where 
$H_{0 L}$ is the Hamiltonian of the harmonic oscillator, arising for
interaction potential $V=0$. So the partition functions
exist for $H_{L \, r}$ by the Golden-Thompson inequality.   

Similarly to the Hamiltonian in 
equation \eqref{e1.1}, the unitary operators 
$e^{itH_L}$ induce automorphisms
$\balpha_L(t) \doteq \mbox{Ad} \, e^{itH_L}$, $t \in \RR$,
of the algebra $\overline{\bfA}$, and Theorem~\ref{t4.6} 
applies accordingly for any given $L$. 
Moreover, in the limit of large~$L$, 
these automorphisms converge pointwise 
to the original dynamics $\balpha(t)$, $t \in \RR$, 
in the topology on~$\overline{\bfA}$ determined by the seminorms.
Proofs of these statement are given in the appendix. 

For the construction of equilibrium states 
in the thermodynamic limit, one picks any $L > 0$
and considers for given inverse temperature $\beta > 0$ 
and chemical potential 
$\mu \leq - V(0)$ the family of states on $\overline{\bfA}$ 
$$
\omega_{\beta, \mu, L}(\, \cdot \,) \doteq
\mbox{Tr} \, (e^{-\beta(H_L - \mu N)} \, \cdot \,) / 
\mbox{Tr} \, e^{-\beta(H_L - \mu N)} \, .
$$
These states satisfy the KMS-condition for the 
dynamics $\balpha_L(t)$, $t \in \RR$, so they describe
equilibria \cite{BrRo}. Moreover, as $L$ approaches infinity,
they have weak-*-limit points $\{ \omega_{\beta, \mu} \}$
on $\overline{\bfA}$ by the Banach-Alaoglu Theorem. 
The limit states need neither be unique,  
as is typically the case in the 
presence of phase transitions. Nor need they 
describe pure 
phases, such as in the presence of spontaneous breakdown
of symmetries, where mixtures of phases can appear.  

It is also not clear from the outset that the limit states describe 
equilibria for the dynamics $\balpha(t)$, $t \in \RR$. 
In case of no interaction, one can show that the limit states 
obtained this way agree with the familiar quasi-free KMS
states obtained in the thermodynamic limit,
cf.~\mbox{\cite[Sec.~5.2.5]{BrRo}}.
A description of quasi-free states in the framework of 
the resolvent algebra can be found in~\cite[Sec.~4]{BuGr2}.
In the presence of interaction some more detailed 
analysis of the limit states is required, however. 
There one can rely again on  
methods developed for C*-dynamical system, such as  
$(\overline{\bfA}_\alpha, \balpha)$, cf.~\cite{BrRo} and
references quoted there. 

The present formalism, describing observables in infinite space, 
provides also a basis for the
discussion of the effects of perturbations of KMS states and their 
return to equilibrium. Moreover, it ought to cover states where 
the translation 
symmetry is spontanteously broken, such as crystals, or 
states in motion, such as fluids. There exists 
an extensive literature on these topics and we refrain
from giving references here; comprehensive lists may be 
found in \cite{BrRo, PiSa}.

We conclude this section by noting that one can also 
study in the present framework the formation of 
Bose-Einstein condensates, trapped by a 
harmonic potential, cf. \cite{CoDeZi,LiSeSoYn,Ve}. There 
arises the interesting question whether 
the condensation phenomenon disappears if the 
external potential is turned off, \ie whether the limit states  
agree on $(\overline{\bfA}_\alpha, \balpha)$ with the 
Fock vacuum. This is likely to be the case in the 
Gross-Pitaevskii model because of the assumed repulsive 
interaction. Yet 
for interaction potentials of positive type, admitting 
bound states, some more interesting limit states  
may appear. 

\section{Summary and outlook}
\setcounter{equation}{0}

In the present article we have established a 
consistent algebraic framework for the treatment of 
interacting non-relativistic Bose fields.
The novel feature of our approach consists of the fact that 
''large field problems'' are avoided from the outset by 
dealing with the resolvents of the fields. Thus in 
singular states, describing accumulations of 
particles with infinite density,  
these operators simply vanish. On the mathematical
side, this implies that the resolvent algebra has 
ideals; but this is inevitable if a   
kinematical algebra is to admit a sufficiently
rich family of different dynamics \cite[Sec.~10.18]{BuGr2}. 
As a consequence of our approach, the 
algebra of gauge invariant observables $\obfA$ 
attains a mathematically convenient structure, 
being the inverse limit of a family of approximately finite
dimensional C*-algebras. The underlying inverse maps
describe the effect of the removal of 
a particle on the structure of the observables
and their dynamics.

The dynamics 
considered in this article, describing attractive
and repulsive interaction potentials, 
act by automorphisms on the algebra of observables $\obfA$ 
in a (with regard to the topology fixed by a natural 
set of seminorms) 
continuous and quasi-local manner. The class of 
admissible potentials can be extended with some additional 
effort to singular two-body interactions and unbounded 
external potentials. Moreover, for given dynamics $\balpha$, one 
can proceed to C*-dynamical systems $(\bfA_\alpha, \balpha)$, where 
$\bfA_\alpha$ is a dense subalgebra of $\obfA$ in the
seminorm topology. 
So interacting Bosons can be described by C*-dynamical 
systems, contrary to obstructions 
conceived in the literature \cite{BrRo,NaTh}. 
 
The present results also shed light on the relation between the 
field theoretic setting and the quantum mechanical 
particle picture, which is based on position and momentum operators. 
Whereas the field-theoretic observable algebra is built from  
operators of finite rank, the particle algebra generated
by position and momentum operators accomodates an
abundance of operators with continuous spectrum. Thus from 
the present point of
view it seems advantageous to first construct the states of 
interest in the field theoretic setting and only then turn
to their analysis and physical interpretation.

In order to illustrate this idea, we have briefly discussed some 
standard problems in many body theory. We have sketched how 
the existence of ground states,
including approximate ground states (infra vacua), can 
be established in the present
framework. Particle features are  
uncovered by proceeding to asymptotic times.
Suitable spatial averages of localized observables converge 
in this limit 
to operators which describe the asymptotic particle momenta.
In a similar manner one can also compute 
collision cross sections, \textit{etc}. 

Finally, we have addressed 
the problem of constructing thermal equilibrium states
for pair potentials of positive type, which include   
potentials allowing for bound states. There 
one is profitting from the fact that the dynamics of a thermal
system, which is trapped by an external harmonic  
potential, also acts by automorphisms on the algebra of 
observables. Turning off the external potential, 
the automorphisms converge on all observables to the 
dynamics of the infinite system. 
Moreover, thermodynamic limits of the trapped equilibrium states exist. 
Yet these limit states may neither be unique nor 
need they necessarily describe equilibria.
Some more detailed analysis is required 
in order to clarify their 
specific properties in each particular case.  

Another topic of interest is the extension of 
the present analysis to the 
non-gauge invariant Bose fields. It seems 
appropriate to proceed there to the C*-algebra containing, 
besides $\obfA$,  
the operators $F_f \doteq a^*(f) \, (1 + a^*(f) a(f))^{-1/2} $, 
$f \in \cD(\RR^s)$, and their adjoints. 
These operators are isometries if $f$ 
is normalized. They are limits of the operators 
$F_{f, \kappa} \doteq a^*(f) \, (1 + a^*(f) a(f))^{-\kappa} $,
$\kappa > 1/2$, which are contained in the resolvent algebra of fields.
Gauge invariant combinations of the operators of the form 
$F_f F_g^*$ and $F_f^* F_g$ already 
appear as elements of the algebra of observables
$\overline{\bfA}$. We conjecture 
that the partially time translated operators 
$\balpha(t)(F_f) \, F_g^*$  and
$F_f \, \balpha(t)(F_g^*)$, 
$t \in \RR$, are also elements of $\obfA$.
The obvious relations $\balpha(t)(F_f) = (\balpha(t)(F_f) F_f^*) \, F_f$ 
and $\balpha(t)(F_f^*) =  F_f^* \, (F_f \balpha(t)(F_f^*))$ 
then imply that the C*-algebra, generated by 
$\obfA$ and the operators $F_f$, \mbox{$f \in \cD(\RR^s)$},  
is stable under the dynamics. We will return to this
problem in a future publication.

Let us mention in conclusion that the present ideas can also 
be applied to the algebra of observables generated by Fermi fields. 
It was shown by Bratteli \cite{Br} that
this algebra has a structure which is similar to that of 
bosonic systems, established here. 
So the present ideas, regarding the dynamics, can be carried over there.
The dynamics of Fermi fields was already studied by Narnhofer 
and Thirring in~\cite{NaTh}. For technical reasons, 
these authors restricted attention to pair potentials with an 
ultraviolet cutoff. It seems that one can 
proceed also there to the generic class of pair
potentials, considered in the present investigation.

\section*{Appendix}
\setcounter{equation}{0}

We establish here results stated in the main text 
in the proof of Lemma~\ref{l4.1} and in Subsection~\ref{s5.3}. 
Given $n \in \NN$, we consider on  
the unsymmetrized Hilbert space $\cH_n$ the 
Hamiltonians 
$$
H_{L n} \doteq \sum_i (\biP_i^2 + \biQ_i^2/L^4) + 
\sum_{j \neq k} V(\biQ_j - \biQ_k)
$$
for $L > 0$, respectively $L = \infty$ (no external forces). 
These Hamiltonians augment the Hamiltonian $H_n$, 
given in equation~\eqref{e2.1}, by an external harmonic potential, where 
we identify $H_{n \, \infty} = H_n$. 
It follows from standard results, cf.\ \cite[Sec.~VIII.7]{ReSi},
that $H_{L n} \rightarrow H_n$ in the strong resolvent sense as
$L \rightarrow \infty$ and consequently 
$e^{it H_{L n}} \rightarrow e^{itH_n}$ in the strong operator topology,
uniformly on compact subsets of $t \in \RR$. 

The adjoint action of the 
unitaries $e^{itH_{L n}}$ on the algebra $\cB(\cH_n)$ is denoted by 
$\alpha_{L n}(t) \doteq \mbox{Ad} \, e^{itH_{L n}}$, where we identify  
$\alpha_{n \infty}(t) = \alpha_n(t)$, $t \in \RR$. 
Each group of automorphisms $\alpha_{L n}(t)$, $t \in \RR$, leaves the 
algebra $\fK(\II_n) \subset \cB(\cH_n)$ invariant and 
acts pointwise norm-continuously
on it. This is apparent for the dynamics 
$\alpha_{L n}^{(0)}(t)$, $t \in \RR$, where the interaction potential 
is put equal to $V = 0$. In order to see that this assertion is also true 
in the interacting case, we proceed as in the main text and
consider the automorphisms  
$\gamma_{L n}(t) \doteq \alpha_{L n}^{(0)}(t) \, \scirc \, \alpha_{L n}(-t)$,
$t \in \RR$, and their expansion in terms of 
multiple integrals as in equation \eqref{e4.2}. 
The basic ingredients in this expansion are the derivations 
and their primitives, $s, t \in \RR$,   
\begin{equation} \label{eA.1}
\bdelta_{L n}(s) \doteq [B, \biV_{\! L n}(s)] \, , \quad
\bDelta_{L n}(t)(B) \doteq \int_0^t \! ds \, 
\bdelta_{L n}(s)(B)  \, , \quad B \in \cB(\cH_n) \, , \tag{A.1}
\end{equation}
cf.\ Lemma \ref{l4.1}. 
Here \ $\biV_{\, L n}(s) \doteq \balpha_{L n}^{(0)}(s)(\biV_n)$,
where \ $\biV_n = \sum_{j \neq k} V(\biQ_j - \biQ_k)$.
Making use of the notation in the proof of that lemma, 
we put \ $\biQ_{j \frown k} \doteq \biQ_j - \biQ_k$, 
$\biP_{j \frown k} \doteq \biP_j - \biP_k$. For the 
individual terms appearing in the 
time shifted sum we then obtain  
$$
V_{L \, j k}(s) \doteq \balpha_{L n}^{(0)}(s)(V(\biQ_{j \frown k}))
= V(c_L(s) \, \biQ_{j \frown k} + s_L(s) \, \biP_{j \frown k}) 
\, \in \, \fR_n(j \frown k) \, .
$$
Here $\, s \mapsto c_L(s) \doteq \cos(2s/L^2)$, 
 $\, s \mapsto s_L(s) \doteq L^2 \, \sin(2s/L^2)$ for finite $L$
and \mbox{$s \mapsto c_\infty(s) \doteq 1$}, 
$s \mapsto s_\infty(s) \doteq 2s$ for $L = \infty$. 
The derivations $\bDelta_{L n}(t)$ in equation~\eqref{eA.1}
are the sum of derivations 
$\Delta_{L \, jk}(t)$, where the potential 
$\biV_{L n}(s)$  in equation~\eqref{eA.1} is replaced by
the pair potential $V_{L \, jk}(s)$.

As was explained in Lemma \ref{l4.1}, an important step 
in the analysis of the derivations $\Delta_{L \, jk}(t)$
consists of the demonstration that the underlying 
integrals $\int_0^t \! ds V_{L \, jk}(s)$, defined in the 
strong operator topology,  belong to the compact
ideal $\fC_n(j \frown k) \subset \fR_n(j \frown k)$. 
Since $\fR_{n}(j \frown k)$ is faithfully represented 
on $L^2(\RR^s)$, it suffices to 
show that the integrals act there as compact operators. 

Since the two-body potential $V$ is an element of $C_0(\RR^s)$, 
the time shifted products $V_{L \, jk}(s') V_{L \, jk}(s'')$ are compact 
operators on the representation 
space $L^2(\RR^s)$ whenever $s',s'' \in \RR$ satisfy 
$h_L(s',s'') \doteq
(s_{L}(s') c_{L}(s'') - s_{L}(s'') c_{L}(s')) \neq 0$. For then the 
two underlying time shifted position operators are canonically 
conjugate with Planck constant $h_{L}(s',s'') \neq 0$. 
Thus, apart from a set of measure zero, 
the function $s',s'' \mapsto V_{L \, jk}(s') V_{L \, jk}(s'')$
on $\RR^2$ has values in the compact operators
on $L^2(\RR^s)$. Moreover, it is bounded in norm.
So the double integral 
$\int_0^t \! ds' \int_0^t \! ds'' \, V_{L \, jk}(s') V_{L \, jk}(s'')
= \big( \int_0^t \! ds \, V_{L \, jk}(s) \big)^2$, defined in 
the strong operator topology, 
is a compact operator for any $t \in \RR$. 
Hence 
its square root is also compact and by polar decomposition 
we arrive at the conclusion that $\int_0^t \! ds \, V_{L jk}(s)$ is 
a compact operator for any $t \in \RR$ and any choice of~$L$. 

By the arguments worked out in the proof of 
Lemma~\ref{l4.1}, one can show then that the algebra 
$\fK(\II_n)$ is stable under the action of 
the derivations $\bDelta_{L n}(t)$, $t \in \RR$.
Moreover, there exists an appropriate generalization 
of  Lemma~\ref{l4.2}, cf.\ the remark that follows it.  
Thus, by the inductive construction of the Dyson 
series, one arrives at a generalization of Proposition~\ref{p4.4}.  
Moreover, the arguments establishing Lemma~\ref{l4.5} can 
still be applied since the interaction
potential $V$ is not changed and vanishes asymptotically.
Hence also the modified dynamics satisfies the coherence condition
$\kappa_n \scirc \, \alpha_{L n}(t) = \alpha_{n-1 \, L}(t) \scirc \, \kappa_n$,
$n \in \NN$. It implies that the algebra $\obfA$ is stable under
the action of the global dynamics $\balpha_L(t)$, $t \in \RR$,
and the statements of Theorem~\ref{t4.6} remain unchanged
for any $L > 0$.

Next, we adress the question regarding the convergence properties of 
the dynamics $\alpha_{L n}(t)$, $t \in \RR$, in the limit of large $L$.
In the non-interacting case, $V=0$, the corresponding 
automorphisms $\alpha_{L n}^{(0)}(t)$ leave each subalgebra 
$\fC_n(\II_m) \subset \fK(\II_n)$ invariant since they do not
mix tensor factors, 
cf.\ the discussion before relation~\ref{e4.1}.
It therefore follows from the convergence properties 
of the underlying unitary operators $e^{itH_{0 \,L n}(t)}$, mentioned above, that 
$\alpha_{L n}^{(0)}(t) \rightarrow \alpha_{n}^{(0)}(t)$
pointwise in norm on $\fK(\II_n)$ if $L$ tends to infinity, $t \in \RR$. 

In order to establish this fact also for the interacting dynamics,
we need to compare the expansions \eqref{e4.2} 
of the Dyson maps 
for different choices of $L$. This is accomplished by 
analyzing the norm distance between the operators 
$\int_0^t \! du V_{L \, jk}(s)$, $t \in \RR$,  for different values of 
$L$. It requires a refinement of the preceding arguments.
We consider again the function 

\vspace*{-5mm}
$$
s',s'' \mapsto V_{L \, jk}(s') V_{L \, jk}(s'') =
e^{is'H_{1 L}} V(\biQ_{j\frown k}) \, e^{i(s''- s')H_{1 L}} \,
V(\biQ_{j\frown k}) \, e^{-is''H_{1 L}} \, ,
$$
where we put $H_{1 L} \doteq \biP_{j\frown k}^2 + \biQ_{j\frown k}^2/L^4$. 
Adopting the Dirac bra-ket notation, the kernel of the middle term 
of this operator is given in configuration space by 
\begin{align*} \label{eA.2} 
\bix, & \biy  \rightarrow 
\langle \bix | \, V(\biQ) \, e^{i(s''-s')H_{1 L}} \, V(\biQ) \, | 
\biy \rangle \tag{A.2}  \\[0.5mm] 
& = i N_L(s'-s'') \, V(\bix) \,  
  e^{i (\big(\sbix^2 + \sbiy^2)\cos(2(s'-s'')/L^2) - 2 \sbix \sbiy
  \big)/2L^2\sin(2(s'-s'')/L^2)}
\ V(\biy) \, .
\end{align*}
Here $N_L(s'-s'') = \big(2 \pi i L^2 \sin(2(s'-s'')/L^2)\big)^{-s/2}$ and
we made use of the Green's function (Mehler kernel) of the 
Hamiltonian $H_{1 L}$. Choosing temporarily for the potential $V$ some 
test function, the kernel \eqref{eA.2} is square integrable   
if the pair $(s',s'')$ lies in the region of regularity 
$$
\biR_L \doteq \{(s',s'') \in \RR^2 : 2(s' - s'') \not\in \pi L^2 \, \ZZ \} \, . 
$$ 
Thus the function 
$s',s'' \mapsto V(\biQ) \, e^{i(s''- s')H_{1 L}} \, V(\biQ)$ \ 
on $\biR_L$ has values in the Hilbert-Schmidt class.  
It is continuous in $s',s''$ 
with regard to the Hilbert-Schmidt norm and converges 
in this topology to 
$s',s'' \mapsto V(\biQ) \, e^{i(s''-s')H_{1}} \, V(\biQ)$,
where $H_1 = \biP^2$, in the limit 
$L \rightarrow \infty$, uniformly on any given compact
subset of~$\biR_L$. 

Since any $V \in C_0(\RR^s)$ can be 
approximated in the supremum norm by testfunctions, the preceding statements 
remain true for such potentials if one replaces the terms 
``Hilbert-Schmidt class'' by ``compact operators'' 
and ``Hilbert-Schmidt norm'' by 
``operator norm'' on $L^2(\RR^s)$, denoted by $\| \, \cdot \, \|_1$.
So for any potential $V \in C_0(\RR^s)$, the function 
\begin{equation} \label{eA.3}
\tag{A.3}
s',s'' \mapsto \| \,  V(\biQ) \, e^{i(s''-s')H_{1 L}} \, V(\biQ) - 
V(\biQ) \, e^{i(s''-s')H_{1}} \, V(\biQ)  \|_1 
\end{equation}
is bounded for $(s',s'') \in \RR^2$, continuous 
on $\biR_L$, and it tends to~$0$
in the limit $L \rightarrow \infty$, uniformly on any given 
compact subset of $\biR_L$. Next, since for $s' \neq s''$ 
the operators $V(\biQ) \, e^{i(s''-s')H_{1}} \, V(\biQ)$, involving
the free Hamiltonian, are
compact and $e^{is'H_{1 L}} \rightarrow  e^{is'H_{1}}$
in the strong operator topology for large $L$, the function 
$$
s',s'' \mapsto 
\|(e^{is'H_{1 L}} - e^{is'H_{1}}) V(\biQ) \, e^{i(s''- s')H_{1}} \, V(\biQ)\|_1
$$
has the same properties as the function in \eqref{eA.3}. This 
applies also to the corresponding function of the adjoint operators.
The preceding facts imply 
\begin{align*}
& \lim_{L \rightarrow \infty} 
\| \int_0^t \! ds' \! \int_0^t \! ds'' \, 
(V_{L \, j k}(s') V_{L \, j k}(s'') - V_{j k}(s') V_{j k}(s'') ) \|_1 \\
& \leq  \int_0^t \! ds' \int_0^t \! ds'' \, 
\lim_{L \rightarrow \infty}  
\| (V_{L \, j k}(s') V_{L \, j k}(s'') - V_{j k}(s') V_{j k}(s'') ) \|_1
= 0 
\end{align*}
by the dominated convergence theorem. As a matter of fact, this 
convergence is uniform on compact subsets of $t \in \RR$ 
since for sufficiently large $L$ the
regions $\biR_L$ contain the square $[0,t]\times[0,t]$,
apart from its diagonal. In a similar manner one obtains. 
\begin{align*}
& \lim_{L \rightarrow \infty} \| \int_0^t \! ds' \,
(V_{L \, jk}(s') - V_{jk}(s')) \int_0^t \! ds'' \, V_{jk}(s'') \|_1 =  0 \, , \\
& \lim_{L \rightarrow \infty} \| \int_0^t \! ds' \,
V_{jk}(s') \int_0^t \! ds'' \, (V_{L \, jk}(s'') - V_{jk}(s'')) \|_1 = 0 \, , 
\end{align*}
because one of the integrals 
appearing in these relations is a compact operator
and the other one converges to $0$ in the strong operator topology. 
Putting everything together, we arrive at  
\begin{align*}
\lim_{L \rightarrow \infty} & \,  
\| \int_0^t \! ds \, ( V_{L \, jk}(s) - V_{jk}(s) )  \|_1^2 \\
& = \lim_{L \rightarrow \infty} 
\| \int_0^t \! ds' \! \int_0^t \! ds'' \, (V_{L \, jk}(s') - V_{jk}(s')) 
(V_{L \, jk}(s'') - V_{jk}(s'')) \|_1 = 0 \, ,
\end{align*} 
where the convergence is uniform on compact subsets of $t \in \RR$. 
It follows from these facts that the derivations 
$t \rightarrow \bDelta_{L n}(t)$
on~$\fK(\II_n)$, on which the Dyson series~\eqref{e4.2} are based,
converge for large $L$ to the original derivation 
$t \rightarrow \bDelta_{n}(t)$. More precisely, 
denoting the norm of linear maps on $\fK_n$ by 
${}_n\| \, \cdot \, \|$, one has 
\begin{align*}
{}_n\| \bDelta_{L n}(t) - \bDelta_{n}(t) \| 
& = {}_n\| \sum_{j \neq k} \, \int_0^t \! ds \, 
[\big( V_{L \, j k}(s) - V_{j k}(s) \big), \, \cdot \, \, ] \| \\
& \leq 2 n(n-1) \, \| \int_0^t \! ds \, ( V_{L \, j k}(s) - V_{j k}(s)) \|_1 \, . 
\end{align*}
Hence, $\lim_{L \rightarrow \infty} {}_n\| \bDelta_{L n}(t) - \bDelta_{n}(t) \| = 0$,
uniformly on compact subsets of $t \in \RR$. 

With this information, we can turn now to the discussion of
the Dyson cocycles. Given $K_n \in \fK_n$,
we represent $\gamma_{L n}(t)(K_n)$ by the series \eqref{e4.2}
of multiple integrals. As has been explained, its inductive construction 
produces at $(l+1)$st order terms of the form   
$t \mapsto D_{L \, l+1}(t) = \int_0^t \! ds \, \bdelta_{L n}(s)(D_{L \, l}(s))$.
Making use of this information, we will show that at 
each order $l \in \NN$, one has
$\lim_{L \rightarrow \infty} \|  D_{L \, l}(t)  - D_l(t) \|_n = 0$,
uniformly on compact sets of $t \in \RR$. 
Since the Dyson series converges absolutely in norm, 
uniformly for $L > 0$, this implies 
$$
\lim_{L \rightarrow \infty} \|\gamma_{L n}(t)(K_n) - \gamma_n(t_n)(K_n) \|_n  = 0 \, 
, \quad K_n \in \fK_n \, , \ n \in \NN \, . 
$$
The proof of convergence of the terms in the series is
obtained again by iteration.
For $l = 1$ one has $D_{L \, 1}(t) = \bDelta_{L n}(t)(K_n)$.
Hence, by the preceding results, 
$\lim_{L \rightarrow \infty} \| D_{L \, 1}(t) - D_1(t) \|_n = 0$,
uniformly on compact subsets of \mbox{$t \in \RR$}. 
Assuming that the corresponding statement holds at order $l$, we proceed to 
\begin{align*}
& D_{L \, l+1}(t) - D_{l+1}(t) \\
& = \int_0^t \! ds \, \bdelta_{L n}(s) \big( D_{L \, l}(s) - D_l(s) \big)
+ \int_0^t \! ds \, (\bdelta_{L n}(s) - \bdelta_n(s))(D_l(s)) \, .
\end{align*}
The first integral is bounded by 
$$
\big\| \int_0^t \! ds \, \bdelta_{L n}(s) \big( D_{L \, l}(s) - D_l(s) ) \big\|_n
\leq |t| \, 2n(n-1) \, \| V \| \, 
\sup_{0 \leq s \leq t} \, \| D_{L \, l}(s) - D_l(s) \|_n \, .
$$
Hence according to the induction hypothesis, 
it vanishes in the limit $L \rightarrow \infty$, uniformly on 
compact subsets of \mbox{$t \in \RR$}. Putting 
$\partial\bDelta_{L n}(t) \doteq \int_0^t \! ds \, (\bdelta_{L n}(s) 
- \bdelta_n(s))$, 
the second integral in the above decomposition 
can be presented as norm limit,
$$
\lim_{k \rightarrow \infty} \sum_{j = 1}^k
\big(\partial\bDelta_{L n}(jt/k) - \partial\bDelta_n((j-1)t/k)\big) \, 
(D(jt/k)) \, .
$$ 
According to Lemma \ref{l4.2}, cf.\ also the remark thereafter,
this sequence converges to the integral in the
norm topology of $\fK_n$ for $k \rightarrow \infty$.
What matters here is the fact that this convergence is
uniform for $L > 0$ and $t$ varying in compact subsets of 
$\RR$. This follows from the bound  
${}_n\| (\bdelta_{L n}(s) - \bdelta_n(s)) \|
\leq 4 n(n-1) \, \| V \| $ which holds
uniformly for $L > 0$ and $s \in \RR$. 
It enters in the estimate \eqref{e4.3} in Lemma~\ref{l4.2}. 
Because of the uniform convergence, 
it suffices to look at the individual terms in the approximating sum. 
Now \mbox{$\lim_{L \rightarrow \infty} {}_n\|\partial\bDelta_{L n}(t)\| = 0$},
uniformly on compact sets of $t \in \RR$. 
Thus each term in the sum vanishes in the norm topology of~$\fK_n$,
hence the second integral in the above decomposition also 
vanishes in this topology, uniformly in $t$. 
This shows that the function $t \mapsto D_{L \, l+1}(t)$ has all required
properties, completing the induction. 
So we conclude that the automorphisms $\balpha_L(t)$ 
converge pointwise on $\obfA$ to $\balpha(t)$, \viz.
$$
\lim_{L \rightarrow \infty} 
\| \balpha_L(t)(A) - \balpha(t)(A) \|_n = 0 \, , \quad 
A \in \obfA, \ n \in \NN_0 \, .
$$
This completes the proofs of statements made in the main text.

\vspace{7mm}
\noindent {\Large \bf Acknowledgment} \\[1mm]
I profitted from a stimulating exchange with Jakob Yngvason 
on a preliminary version of this article. I
am also grateful to Hendrik Grundling and Roberto Longo 
for valuable comments and   
I would like to thank the Mathematics Institute of the   
University of G\"ottingen for their generous hospitality.

\end{document}